\documentclass[12pt]{article}

\usepackage{fullpage}
\usepackage{authblk} % Package pour auteurs et affiliations
\usepackage{xcolor}
\usepackage{amsmath}
\usepackage{amssymb}
\usepackage{amsthm}
\usepackage{physics}
\usepackage{caption}
\usepackage{subcaption}
\usepackage{nicefrac}
\usepackage{tikz}
\usetikzlibrary{positioning,calc,intersections}
\usepackage{anyfontsize}
\usepackage{relsize}
\usepackage{braket}
\usepackage{quantikz}
\usepackage[hidelinks]{hyperref} %hidelinks pour pas de boîtes colorées autour des références
\usepackage{enumitem}
\usepackage{float}
\usepackage[maxbibnames=10,sorting=nyt,date=year,bibstyle=authoryearalphabetic,citestyle=alphabetic,giveninits,maxalphanames=3,minalphanames=3,dashed=false]{biblatex}
\DeclareDelimFormat{finalnamedelim}{\addspace\bibstring{and}\space} %Pour enlever la «Oxford comma» dans les auteurs
\AtEveryBibitem{\clearfield{extradate}}
\DeclareSourcemap{
  \maps[datatype=bibtex]{
    \map{
      \step[fieldsource=issue, match=\regexp{\A[0-9]+\Z}, final]
      \step[fieldset=number, origfieldval]
      \step[fieldset=issue, null]
    }
  }
}
\renewbibmacro*{volume+number+eid}{%
  \printfield{volume}%
  \setunit*{\addnbspace}% 
  \printfield{number}%
  \setunit{\addcomma\space}%
  \printfield{eid}}
\DeclareFieldFormat[article]{number}{\mkbibparens{#1}}
\DeclareFieldFormat{titlecase}{%
  \ifboolexpr{
    test {\ifcurrentfield{title}}
  }
    {\MakeSentenceCase*{#1}}% 
    {#1}% 
}

\newtheorem{theorem}{Theorem}[section]
\newtheorem{lemma}[theorem]{Lemma}

\begin{document}

\title{\fontsize{18}{21.6}\selectfont On the power of multipartite entanglement for pseudotelepathy}

\author[1]{Gilles Brassard}
\author[2,3,4]{Xavier Coiteux-Roy}
\author[1]{Rémi Ligez}
\date{\today}

\affil[1]{Département d'informatique et de recherche opérationnelle, Université de Montréal, Montréal, Canada H3C 3J7}
\affil[2]{Department of Computer Science, University of Calgary, Calgary, Canada T2N 1N4}
\affil[3]{Department of Physics and Astronomy, University of Calgary, Calgary, Canada T2N 1N4}
\affil[4]{Institute for Quantum Science and Technology, University of Calgary, Calgary, Canada T2N 1N4}

\maketitle
\vspace{-7.5em}
\begin{center}
  \fontsize{10}{12}\selectfont\href{mailto:brassard@iro.umontreal.ca}{\texttt{brassard@iro.umontreal.ca}} \quad \href{mailto:xavier.coiteuxroy@ucalgary.ca}{\texttt{xavier.coiteuxroy@ucalgary.ca}} \quad \href{mailto:remi.ligez@umontreal.ca}{\texttt{remi.ligez@umontreal.ca}}
\end{center}
\vspace{2em}
\begin{abstract}
As early as 1935, Schr\"odinger recognized \emph{entanglement} as ``not \emph{one} but rather \emph{the} characteristic trait of quantum mechanics, the one that enforces its entire departure from classical lines of thought''. Indeed, most remarkable phenomena in quantum information science, such as quantum computing and quantum teleportation, spring from clever uses of entanglement. Among them, \emph{pseudotelepathy} enables two or more players to win systematically at some cooperative games with no need for communication between them, a restriction that would make the task impossible in a classical world. We~investigate the power of \emph{multipartite} entanglement for pseudotelepathy. Some known games that can be won with tripartite entanglement cannot be won with bipartite entanglement, but they can be won with bipartite \emph{nonsignalling resources} such as the so-called Popescu--Rohrlich nonlocal box. We exhibit a five-player game that can be won with tripartite entanglement, but not with arbitrary bipartite nonsignalling \mbox{resources} even in the presence of arbitrary five-partite classical resources. This illustrates both the power of bipartite nonsignalling resources (over bipartite entanglement) and the even superior power of tripartite entanglement.
\end{abstract}

\section{Introduction}
Entanglement is at the heart of some of the most fascinating phenomena in quantum information science.
As~soon as it was discovered in 1935 by Einstein, Podolsky and Rosen~\mbox{\cite{EinsteinPodolskyRosen1935}},
Schr\"odinger recognized it as ``not \emph{one} but rather \emph{the} characteristic trait of quantum mechanics, the one that enforces its entire departure from classical lines of thought''~\mbox{\cite{Schrodinger1935}}.
Perhaps its most famous applications are quantum computing~\mbox{\cite{Feynman1982,Deutsch1985}},
quantum teleportation~\mbox{\cite{BBCJPW1993}} 
and entanglement-based quantum cryptography~\mbox{\cite{Ekert1991,BennettBrassardMermin1992}},
but it is also crucial in communication complexity~\mbox{\cite{CleveBuhrman1997}}
because prior entanglement makes it possible for two or more physically separated parties to reduce exponentially
the amount of classical communication required between them to compute some functions
of their joint input~\mbox{\cite{BuhrmanCleveWigderson1998}}.
So-called \emph{nonlocal games}\,\footnote{\,A misnomer since these games can be played just as well in a purely local-realistic universe, Bell's theorem~\mbox{\cite{Bell1964}} notwithstanding~\mbox{\cite{Brassard2025}}.}, defined formally in Section~\ref{subsec:Nonlocal_games}, occur
when communication is reduced to zero. However, to achieve this, we must weaken the objective: none of the participants learn the value of the function
(which would be impossible if it depended on both inputs, since entanglement
could otherwise be harnessed to enable instantaneous communication),
yet they each produce outputs that could not be obtained without communication
in a classical setting.
The ultimate form of nonlocal games is known as \emph{pseudo\-telepathy}~\mbox{\cite{Brassard2003}}, a term
coined in 1999 by Tapp \mbox{\cite{Tapp1999}} because it
gives the \emph{illusion} of instantaneous communication between the participants when in
fact no such communication takes place.
This phenomenon arises when the game can be won with certainty by exploiting entanglement
whereas classical strategies would necessarily entail a nonzero probability of failure.

In~this paper, we study the power of \emph{multi}partite entanglement for pseudotelepathy
by answering positively the following question: Can there be a pseudotelepathy game
that genuinely requires multipartite entanglement to be won?
Such a separation is straightforward when only bipartite entanglement (possibly supplemented by multipartite classical correlations) is allowed: we exhibit in Appendix~\ref{app:GHZ-Mermin} a simple pseudotelepathy game that can be won with tripartite entanglement but not with bipartite entanglement. However, we wish to allow arbitrary bipartite resources, provided they do not entail communication between the participants.
Such resources are called \emph{nonsignalling}. This distinction is fundamental because these resources are known to be more powerful than bipartite entanglement for some tasks~\mbox{\cite{BroadbentMethot2006}}; in particular, we also show that the game mentioned above can nevertheless be won with a single bipartite nonsignalling resource, namely the so-called \mbox{PR--box}~\mbox{\cite{PopescuRohrlich1994}} \mbox{described} below.
We~solve our conundrum with a \emph{five}-player pseudotelepathy game
that needs tripartite entanglement to be won, but leave open the question of whether or not
a \emph{three}-player solution might exist.
To the best of our knowledge, this is the first proof using the inflation technique for pseudotelepathy games.
Furthermore, this illustrates again the power of bipartite
nonsignalling resources, compared to bipartite entanglement, and demonstrates the even superior power of tripartite entanglement.

\section{Preliminaries} \label{sec:Prelim}

In this section, we review the definitions of nonlocal games, pseudotelepathy games, and the classes of strategies relevant to these games, as well as the inflation technique. This will set the stage for the following sections.\vspace{1em}

\subsection{Nonlocal games} \label{subsec:Nonlocal_games}

An \emph{$n$-player nonlocal game}, or simply an \emph{$n$-player game}, is defined as follows. Consider $n$~players, each of whom has a finite input set and a finite output set. Throughout this work, when discussing two- or three-player games, we refer to the players using the conventional names Alice, Bob and Charlie. In a \emph{game instance}, each player is asked a question from his input set. Typically, there is also a \emph{promise} regarding the inputs; meaning that certain combinations of inputs across players might not be allowed. Then, each player must reply with an answer from his output set. The players may agree on a strategy before receiving the questions, including the use of shared quantum states and even nonsignalling resources (defined below), but they are not allowed to communicate once the game instance has started.

The \emph{rule of the game} defines which combinations of outputs are accepted given the questions received by each player. The players \emph{win the game instance} if they manage to answer an accepted combination of outputs for their inputs. Furthermore, the players \emph{win the game} if they win with certainty all possible game instances allowed by the promise. In some cases, we are not only interested in whether or not the players win the game, but rather in their \emph{winning probability}, which can be defined in several inequivalent ways. Perhaps the most natural approach is to define the winning probability of a nonlocal game as the probability of success on the worst game instance, while a more prevalent one defines it as the average probability of success, where the average is taken over some specified probability distribution on the inputs. It is most usual, albeit not necessary, that the inputs be distributed uniformly on all game instances that respect the promise, but it is required that all such instances have nonzero probability. Throughout this work, we adopt the first approach whenever the distinction is relevant; nevertheless, for the purpose of determining whether the players win the game, all these definitions are equivalent because the players win the game if and only if their winning probability is~1.

A special class of these games is known as \emph{pseudotelepathy games} \mbox{\cite{Brassard2003,BrassardBroadbentTapp2005}}. These are games for which there exists a strategy using quantum entanglement that wins the game, whereas no classical strategy can. Two well-known examples of nonlocal games are the CHSH game \mbox{\cite{ClauserHorneShimonyHolt1969}} and the GHZ--Mermin game \mbox{\cite{GreenbergerHorneZeilinger1989,Mermin1990}}. While the latter is pseudotelepathic, the former is not. The CHSH game is a simple 2-player nonlocal game: Alice and Bob each receive an input bit \(x,y \in \{0,1\}\) and must output a bit \(a,b \in \{0,1\}\); they win if \(a \oplus b = x \wedge y\) (in other words, \(a \neq b\) if and only if \(x = y = 1\)).  Here, $\oplus$ denotes the logical XOR (eXclusive-OR) operation and $\wedge$ denotes the logical AND operation. \mbox{Similarly}, we use $\vee$ to denote the logical OR operation. The optimal winning probability, when \mbox{restricted} to quantum strategies, is \(\cos^{2}(\nicefrac{\mathlarger{\pi}}{8}) \approx 85\%\) for the CHSH game \mbox{\cite{Tsirelson1980}}, whereas it would be 75\% when restricted to classical strategies. The GHZ--Mermin game is a \mbox{3-player} nonlocal game: Alice, Bob and Charlie each receive an input bit \( x, y, z \in \{0,1\} \) with the promise that \( x \oplus y \oplus z = 0\). Each player must then output a bit \( a, b, c \in \{0,1\}\), and the players win if \( a \oplus b \oplus c = x \vee y \vee z\).  

\subsection{Strategies}
We now provide a more precise characterization of the various strategies the players may use. First, we consider the case in which the players share \emph{classical randomness}: at the start of each game instance, a variable~$\lambda$, which may be continuous, is sampled from a set~$\Lambda$ according to some probability distribution \(\text{P}_{\Lambda}(\cdot)\), independently for each game instance, and given to each player. The shared randomness may include local hidden variables in the sense of Bell \mbox{\cite{Bell1964}}. A player is \emph{classical} if his output is uniquely determined by his input and the shared randomness. This definition is general, as any local randomness that players may wish to use can be included into the shared randomness. We say that a player is \emph{deterministic} if he is classical and makes no use of the shared randomness. Furthermore, the players are said to use a \emph{quantum strategy} if they initially share one (or, equivalently, more than one) quantum state, and their output is the result of a quantum measurement on their share of the state that depends on their input and the shared randomness\,\footnote{\,It may seem that the shared randomness is not useful in the quantum case since it can be obtained locally by measuring parts of the shared quantum state, but it remains relevant when the quantum states are not shared amongst all players, as we shall consider later.}. Common examples of such states include the bipartite Einstein--Podolsky--Rosen (EPR) state~\mbox{\cite{EinsteinPodolskyRosen1935}} and the tripartite Greenberger--Horne--Zeilinger (GHZ) state~\mbox{\cite{GreenbergerHorneZeilinger1989}}, defined as
\begin{align}
    \ket{\text{EPR}} &= \textstyle{\frac{1}{\sqrt{2}}} \left( \ket{00} + \ket {11}\right) \,,\\
    \ket{\text{GHZ}} &= \textstyle{\frac{1} {\sqrt{2}}} \left( \ket{000} + \ket {111}\right) \,.
\end{align}

More generally, the players may use an arbitrary \emph{nonsignalling strategy}. Intuitively, nonsignalling strategies capture the most general form of correlation that is consistent with the fact that the players cannot communicate after receiving their inputs. They include, but are not limited to, classical and quantum strategies. A nonsignalling strategy is one in which the players' outputs may depend on the combination of inputs and the shared randomness, but in such a way that no player’s input can influence the output distribution of any set of other players. In particular, for every player, the marginal distribution of his output must be independent of the inputs given to the other players. The \emph{PR--box}, named after Popescu and Rohrlich \mbox{\cite{PopescuRohrlich1994}}, is perhaps the best-known example of a nonsignalling correlation that cannot be realized within quantum theory. It is a bipartite device in which each player inputs a bit \( x, y \in \{0,1\} \), and receives a bit \( a, b \in \{0,1\} \) as output, such that \( a \oplus b = x \wedge y \), with both valid pairs of outputs occurring with probability~\( \nicefrac{1}{2} \). This box would obviously allow players to win the CHSH game described above.

\subsection{Inflation}\label{subsec:Inflation}

The \emph{inflation technique} \mbox{\cite{WolfeSpekkensFritz2019}} is a tool developed to analyse correlations in nonlocal games and in more general causal network settings. It is a thought experiment in which multiple copies of the original players and resources are considered. These copies are arranged in a way that preserves the causal structure while making it possible to derive constraints on the original scenario. The main idea is that if a given correlation can arise in the original scenario, then it should also be compatible with the existence of a more general correlation in the inflated one. It follows that the existence of a behaviour in the inflated scenario imposes necessary constraints on the distributions that can be realized in the original one. A~notable advantage of the inflation technique is that it relies on minimal assumptions. The~first is the \emph{device replication} assumption, which states that identical and independent copies of any player and any resource can exist. The second is that all resources considered must be nonsignalling. The following two facts can be derived directly from these assumptions:
\begin{enumerate}[label= Fact \upshape(\Roman*), leftmargin=*]
    \item\label{Fact:1} If two players do not share a common \mbox{resource}, they exhibit statistically independent behaviours.
    \item\label{Fact:2} If two groups of players and their associated resources — whether in the original or inflated scenario — are locally isomorphic, they exhibit identical behaviours.
\end{enumerate}

We present an elementary example, namely the \emph{3-way $p$--coin flip}, illustrating how the inflation technique can be used as a proof method. We say that a group of three players achieves the 3-way $p$--coin flip if their 1-bit outputs are jointly (0,0,0) with probability $p$ and~(1,1,1) with complementary probability \(1-p\):
\begin{equation} \label{eq:ProbCoinFlip}
    \Pr[A = B = C = 0] = p \,;~\Pr[A = B = C = 1] = 1-p \,.    
\end{equation}
We show, via the inflation technique, that for any probability $p$ different from 0 or 1, the 3-way $p$--coin flip cannot be achieved using only bipartite resources, even with access to arbitrarily nonsignalling ones. Assume for a contradiction that three players achieve the 3-way $p$--coin flip with only bipartite resources as depicted in Figure~\ref{fig:Inflation1a} (see next page). We~then use the inflated scenario of Figure~\ref{fig:Inflation1b}.

On one hand, by~\ref{Fact:2}, we get
\begin{equation}
    \Pr\bigl[\hat{A} = \hat{C}\bigr] = \Pr\bigl[A = C\bigr] = 1 \,\text{ and }\Pr\bigl[\hat{B} = \hat{C}\bigr] = \Pr\bigl[B = C\bigr] = 1 \,,
\end{equation}
which, by transitivity, implies that \(\Pr\bigl[\hat{A} = \hat{B}\bigr] = 1\). On~the other hand, by~\mbox{\ref{Fact:1}}, and then~\mbox{\ref{Fact:2}}, for any \mbox{\(a,b \in \{0,1\}\)},
\begin{equation}
    \Pr\bigl[\hat{A} = a,\hat{B} = b\bigr] = \Pr\bigl[\hat{A} = a\bigr] \cdot \Pr\bigl[\hat{B} = b\bigr] = \Pr\bigl[A = a\bigr] \cdot \Pr\bigl[B = b\bigr] \,, 
\end{equation}
which gives us that
\begin{align}
    \Pr\bigl[\hat{A} = \hat{B}\bigr] &=  \Pr\bigl[A = 0\bigr] \cdot \Pr\bigl[B = 0\bigr] + \Pr\bigl[A = 1\bigr] \cdot \Pr\bigl[B = 1\bigr]\nonumber\\
    &= p^{2} + (1-p)^{2} \neq 1
\end{align}
for any $p$ that is neither 0 nor 1, where the last line follows from the marginal probabilities of Eq.~(\ref{eq:ProbCoinFlip}). We conclude from this contradiction that the 3-way $p$--coin flip, for any \(p \in (0,1)\), cannot be obtained by any nonsignalling bipartite resources. With this result in mind, we introduce a 3-player game called the \emph{equality game}. It is a game with no input\,\footnote{\,To satisfy the definition of a game given in Section~\ref{subsec:Nonlocal_games}, we may assign input sets of size 1 to all players.}~in which each player must output a bit \(a,b,c \in \{0,1\} \) such that \(a = b = c\). Clearly, the equality game can be won even without any shared resource, as the players can simply agree in advance to always output the same fixed bit. However, the previous result on the 3-way $p$--coin flip tells us that if, for some reason, the players are prevented from always replying with a fixed bit, then even access to arbitrary nonsignalling bipartite resources is no longer sufficient to win the equality game. This observation will prove crucial later on. For further details on the inflation technique, the reader is referred to Refs.~\mbox{\cite{WolfeSpekkensFritz2019,NavascuesWolfe2020,CoiteuxRoyWolfeRenou2021-2}}.

\begin{figure}[H]
\centering
\begin{subfigure}[b]{0.40\textwidth}
    \centering
    \vspace{0.23cm} 
    \includegraphics[scale=0.3]{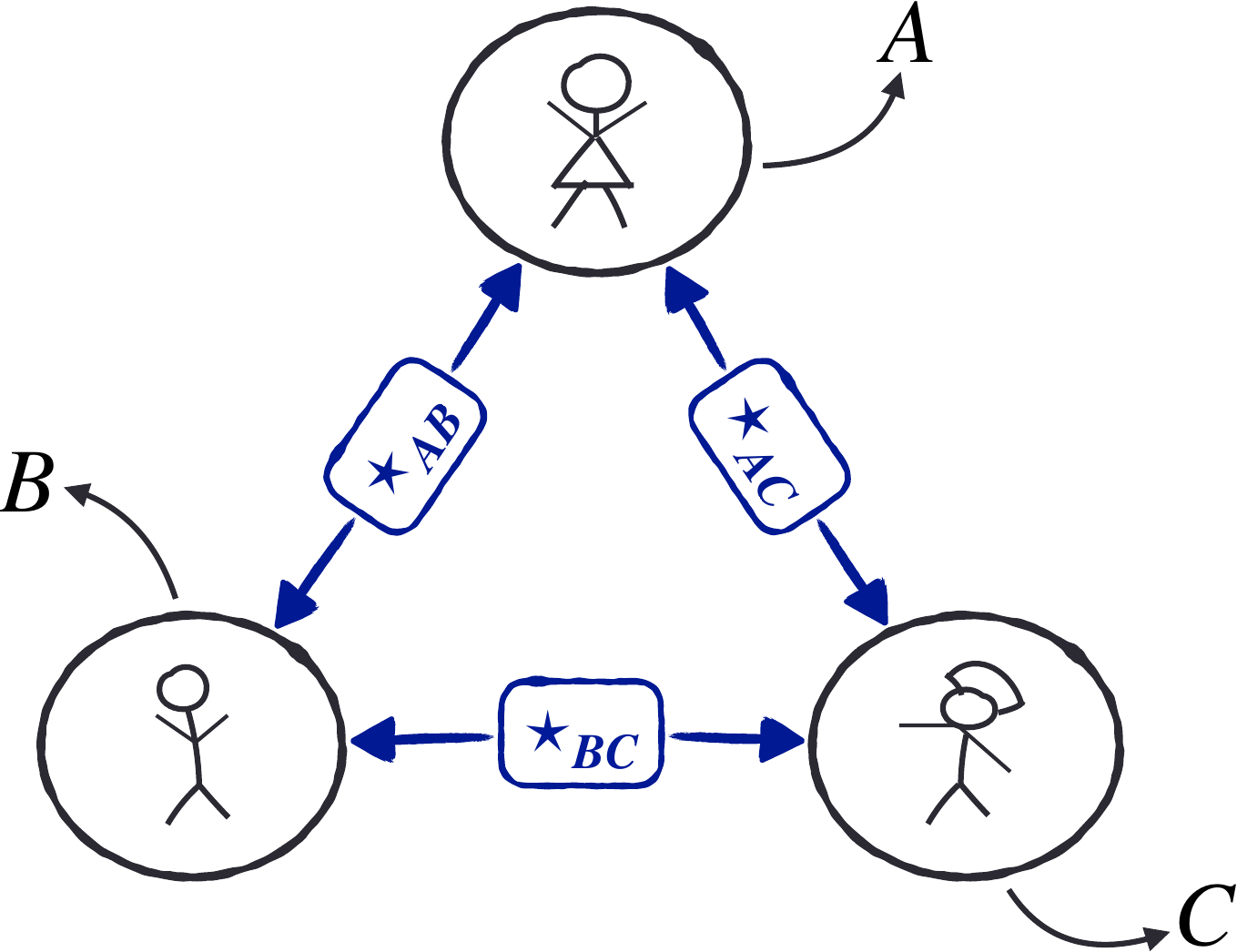}
    \vspace{0.92cm}
    \caption{}
    \label{fig:Inflation1a}
\end{subfigure}
\hfill
\begin{subfigure}[b]{0.56\textwidth}
    \centering
    \includegraphics[scale=0.3]{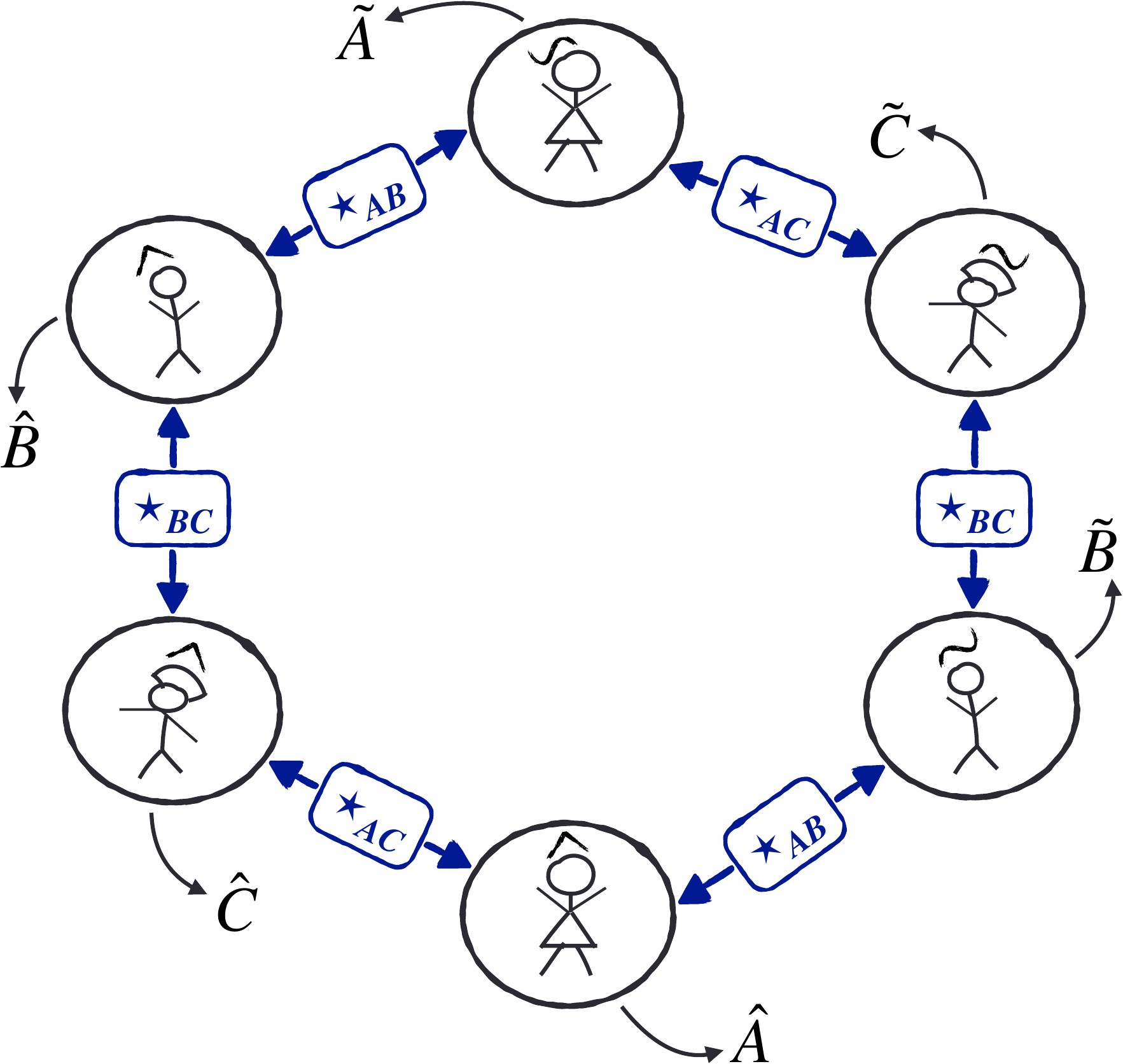}
    \caption{}
    \label{fig:Inflation1b}
\end{subfigure}
\caption{\small (a) Original scenario with three players (Alice, Bob and Charlie) sharing bipartite resources pairwise. (b) Inflated scenario obtained by duplicating and rearranging players and resources. Each duplicated player represents an indistinguishable copy of the same local input-output process as in the original scenario and each duplicated resource represents an independent copy of the
corresponding original resource.}
\label{fig:Inflation1}
\end{figure}

\section{Towards our result}\label{sec:TowardsResult}

In the context of nonlocal games, it is natural to compare the power of different types of shared resources. In particular, we are interested in comparing bipartite and tripartite correlations. If our objective were simply to exhibit a pseudotelepathy game that cannot be won using bipartite quantum resources, the GHZ--Mermin game already provides such an example. It is well known that no classical strategy can win the GHZ--Mermin game, but a quantum strategy in which the players share a single GHZ state can. Although less widely known, it is also true that this game cannot be won using arbitrary bipartite quantum resources, and we provide a proof of this in Appendix~\ref{app:GHZ-Mermin}.
However, this game \emph{can} be won using bipartite nonsignalling resources --- specifically, a single PR--box shared between Alice and Bob. In this case, Alice and Bob input their respective bits into the PR--box and use the outputs as their answers, while Charlie simply outputs his input bit. One can easily verify that this yields a winning strategy for the GHZ--Mermin game. This simple example illustrates why it is relevant to consider not only bipartite quantum correlations, but also bipartite correlations that are arbitrarily nonsignalling.

We now define a new game, the \emph{majority game}, which to the best of our knowledge has not been previously considered in the literature. It is a natural generalization of the GHZ--Mermin game: Alice, Bob and Charlie each receive an input bit \(x,y,z \in \{0,1\}\) with no promise. Each player must then output a bit \(a,b,c \in \{0,1\}\), and the players win if \mbox{\(a \oplus b \oplus c = \text{Maj}(x,y,z)\)}, where \(\text{Maj}(x,y,z) = 1\) if at least two of \(x,y,z\) are 1, and \(\text{Maj}(x,y,z) = 0\) otherwise. This game generalizes the GHZ--Mermin game because \( \text{Maj}(x,y,z) = x \vee y \vee z \) whenever the GHZ--Mermin promise is satisfied. The majority game cannot be won classically --- an immediate consequence of the fact that it generalizes the GHZ--Mermin game, which is pseudotelepathic. However, unlike the GHZ--Mermin game, the majority game is not pseudotelepathic: we prove in Appendix~\ref{app:MajorityGame} that no quantum strategy can win the game with probability greater than \(\cos^{2}(\nicefrac{\mathlarger{\pi}}{8}) \approx 85\%\).

One will recognize this as the optimal quantum winning probability for the CHSH game, and this is no coincidence: as pointed out by Scott Aaronson~\mbox{\cite{Aaronson2025}}, any strategy for the majority game can be transformed into a strategy for CHSH that wins with the same probability.
In Appendix~\ref{app:MajorityGame}, we provide the proof of this reduction, as well as a quantum strategy that achieves this optimal success probability.
Nevertheless, three PR--boxes, one shared by each pair of players, suffice to win the majority game, as we also prove in Appendix~\ref{app:MajorityGame}. This further highlights the remarkable power of PR--boxes \mbox{\cite{BroadbentMethot2006}}, and more generally, of arbitrary nonsignalling correlations: they can surpass quantum theory, achieving success in situations where quantum strategies cannot. For this reason, in what follows, we compare tripartite quantum correlations not just with bipartite quantum ones, but with bipartite correlations that may be arbitrarily nonsignalling. In our setting, these bipartite correlations are also supplemented with multipartite shared randomness.

If our sole objective were to exhibit a game in which tripartite quantum correlations outperform any strategy using only bipartite resources --- even if those resources are \mbox{arbitrarily} nonsignalling --- and shared randomness, Ref.~\mbox{\cite{CoiteuxRoyWolfeRenou2021}} already provided such an example with a three-player game that is essentially a combination of the CHSH game and the equality game. Specifically, Bob and Charlie receive, as part of their inputs, a bit indicating whether Alice and Bob are to play CHSH, or whether all three are to play the equality game, whereas Alice is kept ignorant of this fact. It is proven in Ref.~\mbox{\cite{CoiteuxRoyWolfeRenou2021}} that a higher winning probability can be achieved if Alice, Bob and Charlie share a GHZ state than if they share any bipartite nonsignalling correlations supplemented with tripartite shared randomness. The~intuition is that shared randomness alone cannot win CHSH, whereas bipartite correlations are insufficient to win the equality game. Since Alice does not know which of the two games she is playing, she cannot fully exploit the resources for both. However, this is not a pseudotelepathy game. Recall that CHSH itself is not pseudotelepathic; it is therefore unsurprising that combining the CHSH game with the equality game --- producing a game strictly harder than either one --- results in a game for which no winning quantum strategy exists.

A natural next step is to ask what happens if we replace the CHSH game with a pseudo\-telepathy game game. A~compelling candidate is perhaps the best-known two-player pseudotelepathy game: the \emph{magic square game} (defined below)~\mbox{\cite{Cabello2001,Mermin1990-2,Peres1990}}.~\mbox{However}, this substitution does not lead directly to the desired outcome. Indeed, if one combines the magic square game with the equality game in the same natural manner as in the construction of Ref.~\mbox{\cite{CoiteuxRoyWolfeRenou2021}}, the resulting nonlocal game can in fact be won with a single PR--box, as we show in Appendix~\ref{app:MagicSquareAndEquality} --- yet another striking demonstration of the power of PR--boxes.

We now turn to a different family of pseudotelepathy games that will be central to our result. Specifically, we shall make use of the \emph{magic pentagram game} \mbox{\cite{Mermin1990-2}}. Since its description may not appear intuitive at first sight, we begin by presenting two equivalent formulations of the magic square game: the standard one, and an alternative description that will serve as a stepping stone toward understanding the magic pentagram game. First, the usual way to define the magic square game is based on a \mbox{\(3 \times 3\)} grid whose entries are in~\{0,1\}, as in Figure~\ref{fig:msgGrid}. The objective is to fill the grid so that the parity of each row is \linebreak \vspace{-\baselineskip}
\begin{figure}[H]
    \centering
    \begin{minipage}[t]{0.48\textwidth}
    \centering
    \resizebox{0.68\linewidth}{!}{\begin{tikzpicture}
  \draw[line width=1pt]
    (0,0) grid (3,3);

  \node at (0.5,2.5) {0};
  \node at (1.5,2.5) {0};
  \node at (2.5,2.5) {0};

  \node at (0.5,1.5) {0};
  \node at (1.5,1.5) {1};
  \node at (2.5,1.5) {1};

  \node at (0.5,0.5) {1};
  \node at (1.5,0.5) {0};
  \node at (2.5,0.5) {?};
\end{tikzpicture}}
    \captionsetup{width=\linewidth}
    \captionof{figure}{\small Standard grid formulation of the magic square game. Each cell contains a bit. The parity of each row is required to be even, while the parity of each column is required to be odd (which is impossible).}
    \label{fig:msgGrid}
    \end{minipage}
    \hfill
    \begin{minipage}[t]{0.48\textwidth}
    \centering
    \resizebox{0.68\linewidth}{!}{% Node and edge styles
\tikzset{
  point/.style = {
    circle, draw=black, fill=white,
    line width=1.2pt, minimum size=13mm,
    inner sep=0pt, font=\LARGE
  },
  SolidEdge/.style = {line width=1.6pt},
  DottedEdge/.style = {line width=1.6pt, dashed}
}

\begin{tikzpicture}[node distance=1.8cm and 1.8cm]

% Row 1 (top)
\node[point] (A1) at (0,0)       {0};
\node[point] (A2) [right=of A1]  {0};
\node[point] (A3) [right=of A2]  {0};

% Row 2 (middle)
\node[point] (B1) [below=of A1]  {0};
\node[point] (B2) [below=of A2]  {1};
\node[point] (B3) [below=of A3]  {1};

% Row 3 (bottom)
\node[point] (C1) [below=of B1]  {1};
\node[point] (C2) [below=of B2]  {0};
\node[point] (C3) [below=of B3]  {?};

% Solid horizontal edges (rows)
\foreach \row in {A,B,C} {
  \draw[SolidEdge] (\row1) -- (\row2) -- (\row3);
}

% Dotted vertical edges (columns)
\foreach \col in {1,2,3} {
  \draw[DottedEdge] (A\col) -- (B\col) -- (C\col);
}

\end{tikzpicture}}
    \captionsetup{width=\linewidth}
    \captionof{figure}{\small Hypergraph formulation of the magic square game. Each node corresponds to a cell of the \mbox{$3 \times 3$} grid of Figure~\ref{fig:msgGrid}. Each row or column of the grid of Figure~\ref{fig:msgGrid} defines a hyperedge connecting the corresponding three nodes. Solid lines represent hyperedges whose parity must be even, while dashed lines represent ones whose parity must be odd.}
    \label{fig:msgGraph}    
    \end{minipage}
\end{figure}
\noindent even, while the parity of each column is odd (which is obviously impossible!). Alice is asked to provide the values for all three entries of a specified row, whose parity must be even, while Bob is asked to provide the values for all three entries of a specified column, whose parity must be odd. Furthermore, their answers must be consistent on the single cell where their row and column intersect.

An alternative but equivalent way to describe the magic square game is in terms of a hypergraph (see Figure~\ref{fig:msgGraph}). Each cell of the \mbox{\(3 \times 3\)} grid is viewed as a node of the hypergraph, and each row or column of the grid corresponds to a hyperedge, represented by a line connecting the relevant nodes. The~parity constraint on the rows and columns then becomes a parity constraint on the corresponding hyperedges: every row-hyperedge must have even parity, and every column-hyperedge must have odd parity.

This idea extends naturally to the magic pentagram game. Here, the underlying structure is a pentagram-shaped hypergraph (see Figure~\ref{fig:mpgGraph}). Each of the five straight lines of the pentagram forms a hyperedge containing four nodes and has an assigned parity constraint. Alice and Bob are assigned distinct hyperedges and must provide values for the four nodes of their respective hyperedge. Their answers must satisfy the parity constraints and be consistent on the node where the two hyperedges intersect.
% \begin{figure}[h]
%     \centering
%     \resizebox{0.59\linewidth}{!}{\input{DessinPentagramme.tex}}
%     \caption{\small Hypergraph formulation of the magic pentagram game. Each straight line of the pentagram defines a hyperedge containing four nodes and carrying a parity constraint. Solid lines represent hyperedges whose parity must be even, while the dashed line corresponds to an odd~\mbox{parity}.}
%     \label{fig:mpgGraph}
% \end{figure}
The magic pentagram game~admits a winning quantum strategy, but not a classical one~\mbox{\cite{Mermin1990-2}}. Hence, it is a pseudo\-telepathy game.
For simplicity, instead of requiring the players to provide values for all four nodes of their assigned hyperedge, they will only be asked to provide values for three of them in what follows, as the parity constraint uniquely determines the value of the fourth.
\begin{figure}[H]
    \centering
    \resizebox{0.59\linewidth}{!}{% =======================
% CONFIG — edit labels here
% =======================
% OUTER vertices (P0..P4), counterclockwise starting at top:
%   P0 = top (90°), P1 = 18°, P2 = -54°, P3 = -126°, P4 = 162°
\def\labPzero{0}         % P0 (top)
\def\labPone{?}          % P1 (right-upper)
\def\labPtwo{1}          % P2 (right-lower)
\def\labPthree{1}        % P3 (left-lower)
\def\labPfour{0}         % P4 (left-upper)

% INNER intersection nodes (I0..I4) produced automatically:
\def\labIzero{1}         % I0 (top-right)
\def\labIone{0}          % I1 (middle-right)
\def\labItwo{1}          % I2 (bottom)
\def\labIthree{0}        % I3 (middle-left)
\def\labIfour{1}         % I4 (top-left)

% Node & edge styles
\tikzset{
  point/.style   = {circle, draw=black, fill=white,   % fill=white hides lines behind
                    line width=1.2pt, minimum size=13mm,
                    inner sep=0pt, font=\Large},
  SolidEdge/.style  = {line width=1.6pt},
  DottedEdge/.style = {line width=1.6pt, dashed},
}

\begin{tikzpicture}[scale=1.15, transform shape]

% ---------- Regular pentagon outer vertices ----------
\def\R{5.0}
\path
  ({\R*cos(90) }, {\R*sin(90) })  coordinate (P0)  % top
  ({\R*cos(18) }, {\R*sin(18) })  coordinate (P1)
  ({\R*cos(-54)}, {\R*sin(-54)})  coordinate (P2)
  ({\R*cos(-126)}, {\R*sin(-126)})coordinate (P3)
  ({\R*cos(162)}, {\R*sin(162)})  coordinate (P4);

% ---------- Star edges Pk -> P(k+2) (mod 5) ----------
\path[name path=L0] (P0) -- (P2);
\path[name path=L1] (P1) -- (P3);
\path[name path=L2] (P2) -- (P4);
\path[name path=L3] (P3) -- (P0);
\path[name path=L4] (P4) -- (P1); % horizontal (dotted)

% ---------- Intersections (inner pentagon vertices) ----------
\path[name intersections={of=L0 and L4, by=I0}];
\path[name intersections={of=L0 and L1, by=I1}];
\path[name intersections={of=L1 and L2, by=I2}];
\path[name intersections={of=L2 and L3, by=I3}];
\path[name intersections={of=L3 and L4, by=I4}];

% ---------- Draw edges first ----------
\draw[SolidEdge]  (P0)--(P2);
\draw[SolidEdge]  (P1)--(P3);
\draw[SolidEdge]  (P2)--(P4);
\draw[SolidEdge]  (P3)--(P0);
\draw[DottedEdge] (P4)--(P1); % the horizontal one

% ---------- Draw nodes (white fill hides edges under) ----------
\node[point] at (P0) {\labPzero};
\node[point] at (P1) {\labPone};
\node[point] at (P2) {\labPtwo};
\node[point] at (P3) {\labPthree};
\node[point] at (P4) {\labPfour};

\node[point] at (I0) {\labIzero};
\node[point] at (I1) {\labIone};
\node[point] at (I2) {\labItwo};
\node[point] at (I3) {\labIthree};
\node[point] at (I4) {\labIfour};

% % ---------- DEBUG ----------
% \draw[red, line width=1pt] (current bounding box.south west)
%   rectangle (current bounding box.north east);

% % bounding-box center
% \path let
%   \p1 = (current bounding box.south west),
%   \p2 = (current bounding box.north east)
% in coordinate (BBCenter) at ({(\x1+\x2)/2},{(\y1+\y2)/2});

% \draw[blue, line width=1pt] (BBCenter |- current bounding box.south)
%   -- (BBCenter |- current bounding box.north);

% \draw[blue, line width=1pt] (current bounding box.west |- BBCenter)
%   -- (current bounding box.east |- BBCenter);

% \fill[blue] (BBCenter) circle (2pt);
% \node[blue, above right] at (BBCenter) {box center};

% \fill[green!60!black] (0,0) circle (2pt);
% \node[green!60!black, below right] at (0,0) {$(0,0)$};

\end{tikzpicture}}
    \caption{\small Hypergraph formulation of the magic pentagram game. Each straight line of the pentagram defines a hyperedge containing four nodes and carrying a parity constraint. Solid lines represent hyperedges whose parity must be even, while the dashed line corresponds to an odd~\mbox{parity}.}
    \label{fig:mpgGraph}
\end{figure}
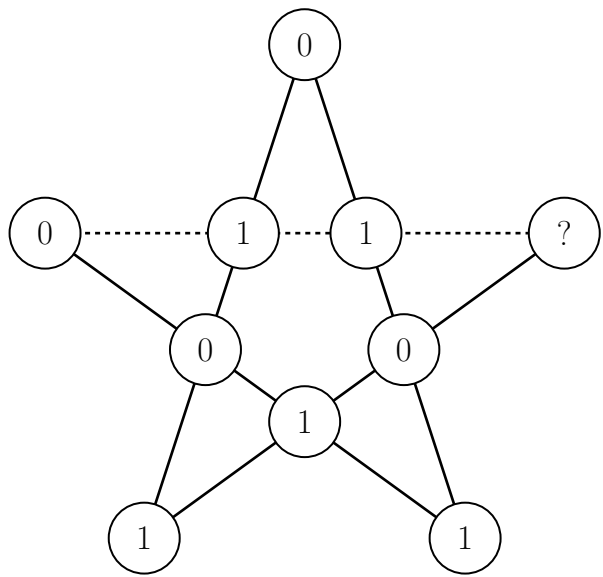

\section{Our result}\label{sec:OurResult}

The central idea of our construction is to combine two types of games: a pseudotelepathy game, which cannot be won using classical strategies, and the equality game, which cannot be won using only bipartite resources, with the objective of producing a game that requires genuinely multipartite entanglement to be won. For the pseudotelepathy game, we shall use a \emph{distributed version} of the magic pentagram game introduced in Section~\ref{sec:TowardsResult}, in which Alice is split into three distinct players denoted $\textnormal{Alice}_{1}$, $\textnormal{Alice}_{2}$ and $\textnormal{Alice}_{3}$. Each of them is responsible for labelling exactly one of the three nodes of a nonhorizontal hyperedge that, in the original game, could have been assigned to Alice as a whole. More formally, for each~$i$, $\textnormal{Alice}_{i}$ receives as input a node \(x_i \in \{0,1\}\) in the pentagram hypergraph according to Figure~\ref{fig:MPGGraph3Alices}, with the promise that $x_1$, $x_2$ and $x_3$ belong to the same nonhorizontal hyperedge (i.e.~\( (x_1,x_2,x_3) \in \{(0,0,0),(0,1,1),(1,0,1),(1,1,0)\}\), or, more compactly, \mbox{\( x_1 \oplus x_2 \oplus x_3 = 0\)}).
% \begin{figure}[h]
%     \centering
%     \resizebox{0.64\linewidth}{!}{\input{DessinPentagramme3Alices.tex}}
%     \caption{\small Distributed version of the magic pentagram game with three Alices. Each $\textnormal{Alice}_i$ is assigned one node of a common hyperedge and receives as input the corresponding node label \mbox{$x_i \in \{0,1\}$}, with the promise that the three inputs belong to the same hyperedge. Bob is assigned the horizontal hyperedge and outputs the labels $b_1,b_2,b_3$ of its three leftmost nodes. Solid lines represent hyperedges whose parity must be even, while the dashed line corresponds to an odd parity. The value of the fourth node of the horizontal hyperedge is uniquely determined by its parity constraint.}
%     \label{fig:MPGGraph3Alices}
% \end{figure}
Bob, on the other hand, receives no input (or, equivalently, a fixed one) as he is always assigned the horizontal hyperedge. Each $\textnormal{Alice}_{i}$ outputs a bit \(a_i \in \{0,1\}\) corre\-sponding to\linebreak

\begin{figure}[H]
    \centering
    \resizebox{0.64\linewidth}{!}{% =======================
% CONFIG — edit labels here
% =======================
% OUTER vertices (P0..P4), counterclockwise starting at top:
%   P0 = top (90°), P1 = 18°, P2 = -54°, P3 = -126°, P4 = 162°
\def\labPzero{$x_1\!=\!1$}   % P0 (top)
\def\labPone{}               % P1 (right-upper)
\def\labPtwo{$x_2\!=\!1$}    % P2 (right-lower)
\def\labPthree{$x_2\!=\!0$}  % P3 (left-lower)
\def\labPfour{$b_1$}         % P4 (left-upper)

% INNER intersection nodes (I0..I4) produced automatically:
\def\labIzero{$b_3$}         % I0 (top-right)
\def\labIone{$x_3\!=\!0$}    % I1 (middle-right)
\def\labItwo{$x_1\!=\!0$}    % I2 (bottom)
\def\labIthree{$x_3\!=\!1$}  % I3 (middle-left)
\def\labIfour{$b_2$}         % I4 (top-left)

% Node & edge styles
\tikzset{
  point/.style   = {circle, draw=black, fill=white,   % fill=white hides lines behind
                    line width=1.2pt, minimum size=16mm,
                    inner sep=0pt, font=\Large},
  SolidEdge/.style  = {line width=1.6pt},
  DottedEdge/.style = {line width=1.6pt, dashed},
}

\centering
\begin{tikzpicture}[scale=1.15, transform shape]

% ---------- Regular pentagon outer vertices ----------
\def\R{5.0}
\path
  ({\R*cos(90) }, {\R*sin(90) })  coordinate (P0)  % top
  ({\R*cos(18) }, {\R*sin(18) })  coordinate (P1)
  ({\R*cos(-54)}, {\R*sin(-54)})  coordinate (P2)
  ({\R*cos(-126)}, {\R*sin(-126)})coordinate (P3)
  ({\R*cos(162)}, {\R*sin(162)})  coordinate (P4);

% ---------- Star edges Pk -> P(k+2) (mod 5) ----------
\path[name path=L0] (P0) -- (P2);
\path[name path=L1] (P1) -- (P3);
\path[name path=L2] (P2) -- (P4);
\path[name path=L3] (P3) -- (P0);
\path[name path=L4] (P4) -- (P1); % horizontal (dotted)

% ---------- Intersections (inner pentagon vertices) ----------
\path[name intersections={of=L0 and L4, by=I0}];
\path[name intersections={of=L0 and L1, by=I1}];
\path[name intersections={of=L1 and L2, by=I2}];
\path[name intersections={of=L2 and L3, by=I3}];
\path[name intersections={of=L3 and L4, by=I4}];

% ---------- Draw edges first ----------
\draw[SolidEdge]  (P0)--(P2);
\draw[SolidEdge]  (P1)--(P3);
\draw[SolidEdge]  (P2)--(P4);
\draw[SolidEdge]  (P3)--(P0);
\draw[DottedEdge] (P4)--(P1); % the horizontal one

% ---------- Draw nodes (white fill hides edges under) ----------
\node[point] at (P0) {\labPzero};
\node[point] at (P1) {\labPone};
\node[point] at (P2) {\labPtwo};
\node[point] at (P3) {\labPthree};
\node[point] at (P4) {\labPfour};

\node[point] at (I0) {\labIzero};
\node[point] at (I1) {\labIone};
\node[point] at (I2) {\labItwo};
\node[point] at (I3) {\labIthree};
\node[point] at (I4) {\labIfour};

% % ---------- DEBUG ----------
% \draw[red, line width=1pt] (current bounding box.south west)
%   rectangle (current bounding box.north east);

% % bounding-box center
% \path let
%   \p1 = (current bounding box.south west),
%   \p2 = (current bounding box.north east)
% in coordinate (BBCenter) at ({(\x1+\x2)/2},{(\y1+\y2)/2});

% \draw[blue, line width=1pt] (BBCenter |- current bounding box.south)
%   -- (BBCenter |- current bounding box.north);

% \draw[blue, line width=1pt] (current bounding box.west |- BBCenter)
%   -- (current bounding box.east |- BBCenter);

% \fill[blue] (BBCenter) circle (2pt);
% \node[blue, above right] at (BBCenter) {box center};

% \fill[green!60!black] (0,0) circle (2pt);
% \node[green!60!black, below right] at (0,0) {$(0,0)$};

\end{tikzpicture}%
}
    \caption{\small Distributed version of the magic pentagram game with three Alices. Each $\textnormal{Alice}_i$ is assigned one node of a common hyperedge and receives as input the corresponding node label \mbox{$x_i \in \{0,1\}$}, with the promise that the three inputs belong to the same hyperedge. Bob is assigned the horizontal hyperedge and outputs the labels $b_1,b_2,b_3$ of its three leftmost nodes. Solid lines represent hyperedges whose parity must be even, while the dashed line corresponds to an odd parity. The value of the fourth node of the horizontal hyperedge is uniquely determined by its parity constraint.}
    \label{fig:MPGGraph3Alices}
\end{figure}
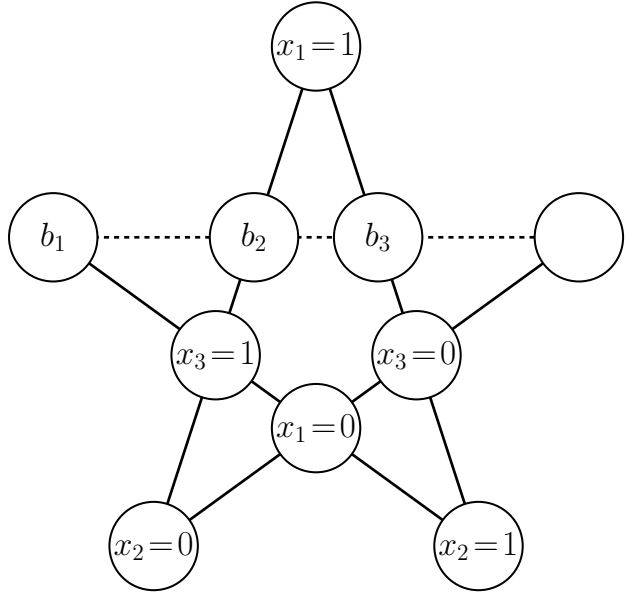

\noindent the label of her assigned node, while Bob outputs three bits $b_1,b_2,b_3$ corresponding to the labels of the three leftmost nodes of the horizontal hyperedge.
To~win, the labels assigned to each hyperedge must satisfy the corresponding parity constraint, with the value of the fourth node of Bob's horizontal hyperedge uniquely determined by its parity constraint.
For example, if the Alices' inputs are \mbox{\( (x_1,x_2,x_3) = (0,0,0) \)}, the winning condition \mbox{becomes}  \mbox{\( a_1 \oplus a_2 \oplus a_3 = b_1 \oplus b_2 \oplus b_3 \oplus 1 \)}.

A~quantum winning strategy for this game can be achieved using three EPR states, each shared between Bob and one of the three Alices.
If \mbox{\( x_i = 0 \)}, $\textnormal{Alice}_{i}$ measures her qubit in the Hadamard basis, whereas she measures it in the computational basis if \mbox{\(x_i = 1\)}.
In either case, her measurement outcome is her answer~$a_{i}$. Bob's strategy is to apply the unitary transformation \textsf{Maj} to his three qubits, where \textsf{Maj} computes the majority \mbox{function: for \( x,y,z \in \{0,1\} \)},
\begin{equation}
    \textsf{Maj}\ket{xyz} = (-1)^{\textnormal{Maj}(x,y,z)}\ket{xyz} \,.
\end{equation}
This transformation can be implemented from elementary gates, with the circuit shown in Figure~\ref{fig:CircuitMaj}. After applying this majority transformation, Bob measures his qubits in the Hadamard basis. The resulting measurement outcomes are his answers \( (b_1, b_2,b_3)\). A tedious but elementary computation shows that this yields a winning strategy.

\begin{figure}[H]
    \centering
    \resizebox{0.72\linewidth}{!}{\input{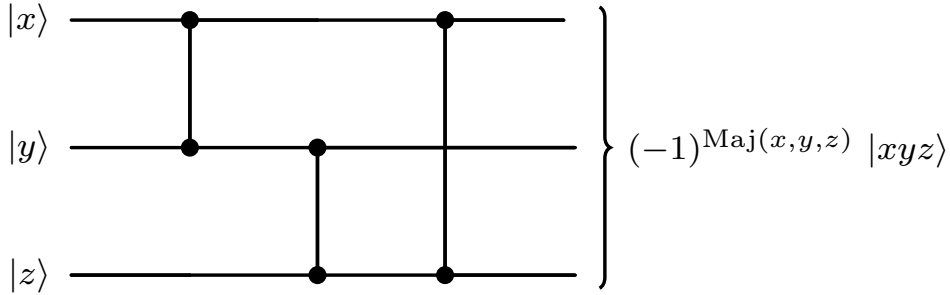}}
    \caption{\small Quantum circuit implementing the unitary \textsf{Maj}. It uses three controlled-Z gates.}
    \label{fig:CircuitMaj}
\end{figure}

However, this game cannot be won classically. Indeed, if a winning classical strategy existed, then by fixing the shared randomness, we would obtain a winning deterministic strategy. Such a strategy would assign a fixed value $0$ or $1$ to each node of the pentagram, but no fixed assignment can satisfy all the parity constraints simultaneously. To see this, observe that each node belongs to exactly two hyperedges, so summing all node values of each hyperedge counts each node twice and must therefore yield an even total. On the other hand, the same sum can also be obtained by adding the parities of the five hyperedges. Since four of the five hyperedges require even parity and one requires odd parity, the total must be odd --- a~contradiction.
Note that this argument remains valid for any hypergraph game for which a winning deterministic strategy assigns a fixed value to each node and in which each node belongs to exactly two hyperedges, provided that the number of hyperedges requiring odd parity is itself odd --- a~fact that we shall use later.

\subsection{The main game}\label{subsec:TheMainGame}

We now define the main game, which involves five players: $\textnormal{Alice}_1$, $\textnormal{Alice}_2$, $\textnormal{Alice}_3$, Bob and Charlie. Each player receives one input bit: $x_1$, $x_2$, $x_3$ for the three Alices and a common bit $z$ for Bob and Charlie. Each $\textnormal{Alice}_i$ must output a bit $a_i$, while Bob and Charlie must each output three bits, $b_1,b_2,b_3$ and $c_1,c_2,c_3$, respectively.

When \mbox{$z = 0$}, the players are to play the hypergraph game represented in Figure~\ref{fig:MPGGraph3AlicesCharlie}, which is based on the distributed version of the magic pentagram game described above, under the same promise that $x_1 \oplus x_2 \oplus x_3 = 0$. The hypergraph is obtained from Figure~\ref{fig:MPGGraph3Alices} by adding a lower horizontal hyperedge, which is assigned to Charlie. Bob's outputs are assigned to the three leftmost nodes of the upper horizontal hyperedge, while Charlie's outputs are assigned to the three rightmost nodes of the lower horizontal hyperedge. Because of the promise, the three nodes assigned to the Alices always lie on a common nonhorizontal hyperedge, as before. To~win, the labels assigned to each hyperedge must satisfy the corresponding parity constraint, with the value of the fourth node of each horizontal hyperedge uniquely determined by its parity constraint.
\begin{figure}[H]
    \centering
    \resizebox{0.67\linewidth}{!}{% =======================
% CONFIG — edit labels here
% =======================
% OUTER vertices (P0..P4), counterclockwise starting at top:
%   P0 = top (90°), P1 = 18°, P2 = -54°, P3 = -126°, P4 = 162°
\def\labPzero{$a_1^{(x_1=1)}$}   % P0 (top)
\def\labPone{}               % P1 (right-upper)
\def\labPtwo{$a_2^{(x_2=1)}$}    % P2 (right-lower)
\def\labPthree{$a_2^{(x_2=0)}$}  % P3 (left-lower)
\def\labPfour{$b_1^{(z=0)}$}         % P4 (left-upper)

% INNER intersection nodes (I0..I4) produced automatically:
\def\labIzero{$b_3^{(z=0)}$}         % I0 (top-right)
\def\labIone{$a_3^{(x_3=0)}$}    % I1 (middle-right)
\def\labItwo{$a_1^{(x_1=0)}$}    % I2 (bottom)
\def\labIthree{$a_3^{(x_3=1)}$}  % I3 (middle-left)
\def\labIfour{$b_2^{(z=0)}$}         % I4 (top-left)

% Node & edge styles
\tikzset{
  point/.style   = {circle, draw=black, fill=white,   % fill=white hides lines behind
                    line width=1.2pt, minimum size=20mm,
                    inner sep=0pt, font=\Large},
  SolidEdge/.style  = {line width=1.6pt},
  DottedEdge/.style = {line width=1.6pt, dashed},
}

\begin{tikzpicture}[scale=1.15, transform shape]

% ---------- Regular pentagon outer vertices ----------
\def\R{5.0}
\path
  ({\R*cos(90) }, {\R*sin(90) })  coordinate (P0)  % top
  ({\R*cos(18) }, {\R*sin(18) })  coordinate (P1)
  ({\R*cos(-54)}, {\R*sin(-54)})  coordinate (P2)
  ({\R*cos(-126)}, {\R*sin(-126)})coordinate (P3)
  ({\R*cos(162)}, {\R*sin(162)})  coordinate (P4);

% ---------- Star edges Pk -> P(k+2) (mod 5) ----------
\path[overlay, name path=L0] (P0) -- (P2);
\path[overlay, name path=L1] (P1) -- (P3);
\path[overlay, name path=L2] (P2) -- (P4);
\path[overlay, name path=L3] (P3) -- (P0);
\path[overlay, name path=L4] (P4) -- (P1); % horizontal (dotted)

% ---------- Intersections (inner pentagon vertices) ----------
\path[overlay, name intersections={of=L0 and L4, by=I0}];
\path[overlay, name intersections={of=L0 and L1, by=I1}];
\path[overlay, name intersections={of=L1 and L2, by=I2}];
\path[overlay, name intersections={of=L2 and L3, by=I3}];
\path[overlay, name intersections={of=L3 and L4, by=I4}];

% ---------- Draw edges first ----------
% \draw[SolidEdge]  (P0)--(P2);
% \draw[SolidEdge]  (P1)--(P3);
% \draw[SolidEdge]  (P2)--(P4);
% \draw[SolidEdge]  (P3)--(P0);
% \draw[DottedEdge] (P4)--(P1); % the horizontal one

% ---------- Draw nodes (white fill hides edges under) ----------
% \node[point] at (P0) {\labPzero};
% \node[point] at (P1) {\labPone};
% \node[point] at (P2) {\labPtwo};
% \node[point] at (P3) {\labPthree};
% \node[point] at (P4) {\labPfour};

% \node[point] at (I0) {\labIzero};
% \node[point] at (I1) {\labIone};
% \node[point] at (I2) {\labItwo};
% \node[point] at (I3) {\labIthree};
% \node[point] at (I4) {\labIfour};

% ---------- New bottom nodes ----------
% Requires: \usetikzlibrary{calc,intersections}

\def\Ybottom{-6.5} % adjust this to move the row up/down

\path[overlay, name path=BottomRow] (-8,\Ybottom) -- (8,\Ybottom);

\def\labQzero{}
\def\labQone{$c_2^{(z=0)}$}
\def\labQtwo{$c_3^{(z=0)}$}
\def\labQthree{$c_1^{(z=0)}$}

% Extended lines
\path[overlay, name path=E0] (I2) -- ($(I2)!2.5!(P3)$); % I2 -> P3
\path[overlay, name path=E1] (I3) -- ($(I3)!2.5!(P3)$); % I3 -> P3
\path[overlay, name path=E2] (I1) -- ($(I1)!2.5!(P2)$); % I1 -> P2
\path[overlay, name path=E3] (I2) -- ($(I2)!2.5!(P2)$); % I2 -> P2

% Intersections with the bottom row
\path[overlay, name intersections={of=E0 and BottomRow, by=Q0}];
\path[overlay, name intersections={of=E1 and BottomRow, by=Q1}];
\path[overlay, name intersections={of=E2 and BottomRow, by=Q2}];
\path[overlay, name intersections={of=E3 and BottomRow, by=Q3}];

% Optional: draw guiding extensions
% \draw[SolidEdge, gray!50] (P3) -- (Q0);
% \draw[SolidEdge, gray!50] (P3) -- (Q1);
% \draw[SolidEdge, gray!50] (P2) -- (Q2);
% \draw[SolidEdge, gray!50] (P2) -- (Q3);

\draw[SolidEdge]  (P0)--(Q2);
\draw[SolidEdge]  (P1)--(Q0);
\draw[SolidEdge]  (P4)--(Q3);
\draw[SolidEdge]  (P0)--(Q1);
\draw[SolidEdge]  (Q0)--(Q3);
\draw[DottedEdge] (P4)--(P1); % the horizontal one

% New nodes

\node[point] at (P0) {\labPzero};
\node[point] at (P1) {\labPone};
\node[point] at (P2) {\labPtwo};
\node[point] at (P3) {\labPthree};
\node[point] at (P4) {\labPfour};

\node[point] at (I0) {\labIzero};
\node[point] at (I1) {\labIone};
\node[point] at (I2) {\labItwo};
\node[point] at (I3) {\labIthree};
\node[point] at (I4) {\labIfour};

\node[point] at (Q0) {\labQzero};
\node[point] at (Q1) {\labQone};
\node[point] at (Q2) {\labQtwo};
\node[point] at (Q3) {\labQthree};

% ---------- DEBUG ----------
% \begin{scope}[overlay]
%     \draw[red, line width=1pt] (current bounding box.south west) rectangle (current bounding box.north east);
% \end{scope}

% % bounding-box center
% \path let
%   \p1 = (current bounding box.south west),
%   \p2 = (current bounding box.north east)
% in coordinate (BBCenter) at ({(\x1+\x2)/2},{(\y1+\y2)/2});

% \draw[blue, line width=1pt] (BBCenter |- current bounding box.south)
%   -- (BBCenter |- current bounding box.north);

% \draw[blue, line width=1pt] (current bounding box.west |- BBCenter)
%   -- (current bounding box.east |- BBCenter);

% \fill[blue] (BBCenter) circle (2pt);
% \node[blue, above right] at (BBCenter) {box center};

% \fill[green!60!black] (0,0) circle (2pt);
% \node[green!60!black, below right] at (0,0) {$(0,0)$};

\end{tikzpicture}}
    \caption{\small Hypergraph representation of the $z=0$ case of the main game. Each $\textnormal{Alice}_i$ is assigned one node ($x_i$) of a common nonhorizontal hyperedge and outputs the label $a_i$ of that node. Bob is assigned the upper horizontal hyperedge and outputs the labels $b_1,b_2,b_3$ of its three leftmost nodes, while Charlie is assigned the lower horizontal hyperedge and outputs the labels $c_2,c_3,c_1$ of its three rightmost nodes. Solid lines represent hyperedges whose parity must be even, while the dashed line corresponds to an odd parity. The value of the fourth node of each horizontal hyperedge is uniquely determined by the corresponding parity constraint.}
    \label{fig:MPGGraph3AlicesCharlie}
\end{figure}
As Bob's horizontal hyperedge is the only one with odd parity, the total number of hyperedges with odd parity is itself odd, and therefore the resulting game cannot be won classically, as argued previously.

On the other hand, when \mbox{$z = 1$}, the promise is that \mbox{\( x_1 = x_2 = x_3 = 1 \)}, and all players are instead to play three parallel instances of the equality game: their outputs must satisfy
\begin{equation}
    a_i = b_i = c_i \,, 
\end{equation}
for every $i \in \{1,2,3\}$.

A quantum winning strategy for the main game can be achieved using three GHZ states, each shared between Charlie, Bob and one of the three Alices. Each $\textnormal{Alice}_i$ behaves exactly as in the distributed version of the magic pentagram game described previously.

When \mbox{$z=1$}, both Bob and Charlie measure their qubits in the computational basis and output the measurement outcomes. By the promise, this also implies that \( x_1 = x_2 = x_3 = 1 \), which, in the Alices’ strategy for the distributed version of the magic pentagram game, corresponds to measuring in the computational basis as well. Hence, when \mbox{$z = 1$}, all players measure their GHZ states in the computational basis, ensuring that \( a_i = b_i = c_i \) for all~$i$, and therefore winning the three parallel instances of the equality game.

When \mbox{$z=0$}, Charlie first applies a Hadamard transform to each of his qubits, transforming each GHZ state into 
\begin{equation}
    \textstyle{\frac{1}{\sqrt{2}}} \ket{\text{EPR}}\ket{0} + \textstyle{\frac{1}{\sqrt{2}}}\ket{\text{EPR}^{\prime}}\ket{1} \,,
\end{equation}
where
\begin{equation}
    \ket{\text{EPR}^{\prime}} = \textstyle{\frac{1}{\sqrt{2}}}(\ket{00} - \ket{11})\,.
\end{equation}

Therefore, after Charlie measures his qubits and obtains outputs \(c_1,c_2,c_3\), the pair $(\textnormal{Alice}_i, \textnormal{Bob})$ shares the state $\ket{\textnormal{EPR}}$ when $c_i=0$ and $\ket{\textnormal{EPR}^{\prime}}$ when $c_i=1$.
Recall that, in the strategy of the distributed version of the magic pentagram game (Figure~\ref{fig:MPGGraph3Alices}), each $\textnormal{Alice}_i$ measures in the Hadamard basis if $x_i = 0$, and in the computational basis otherwise. Therefore, replacing $\ket{\textnormal{EPR}}$ by $\ket{\textnormal{EPR}^{\prime}}$ for the pair $(\textnormal{Alice}_i, \textnormal{Bob})$ will introduce a bit flip if and only if $x_i = 0$.
Now observe that each nonhorizontal hyperedge, whether in Figure~\ref{fig:MPGGraph3Alices} or in Figure~\ref{fig:MPGGraph3AlicesCharlie}, contains either one or three nodes labelled by \mbox{$x_i=0$}. Moreover, in Figure~\ref{fig:MPGGraph3AlicesCharlie}, whenever a nonhorizontal hyperedge contains the node with Charlie's output $c_i$, it also contains the node labelled by \mbox{$x_i=0$}. Furthermore, the nonhorizontal hyperedge containing the leftmost node of the lower horizontal hyperedge, whose value is given by $c_1\oplus c_2\oplus c_3$ by the parity constraint, contains all three nodes labelled by $x_i=0$. Consequently, whenever $c_i=1$, the induced bit flip on $\textnormal{Alice}_i$'s output appears on the nonhorizontal hyperedge that contains the node $c_i$, as well as on the nonhorizontal hyperedge that contains the node whose value is given by $c_1 \oplus c_2 \oplus c_3$. Hence, the bit flips induced on the Alices' outputs are exactly compensated by Charlie's outputs on the lower horizontal hyperedge. 
The Alices and Bob can execute their strategy for the distributed version of the magic pentagram game, which, crucially, does not depend explicitly on Charlie’s measurement outcomes. Together with Charlie's outputs on the lower horizontal hyperedge, this satisfies all the parity constraints of Figure~\ref{fig:MPGGraph3AlicesCharlie}. Hence, this yields a quantum winning strategy using tripartite entanglement for the main game.
    
We have defined a pseudotelepathy game that admits a quantum winning strategy relying on tripartite entanglement. In the next section, we prove that no strategy using arbitrary nonsignalling bipartite correlations, even when supplemented with multipartite shared randomness, can win this game.

\section{The proof}
We now prove the aforementioned result, together with a lemma that will assist in its proof, which admits a more general formulation that may be of independent interest. First, we~recall a well-known fact about nonlocal games that admit a winning strategy.
\begin{lemma}\label{lem:SR}
    If there exists a winning strategy for a nonlocal game, then there also exists one that uses no shared randomness, but otherwise all the same resources.
\end{lemma}
For completeness, we include the proof in Appendix~\ref{app:LemmasNonlocalGames}. We now introduce a new lemma, which will play a key role in our proof.
\begin{lemma}\label{lem:InputDeter}
    If there exists a nonsignalling strategy for a game in which each player has only two possible inputs and behaves deterministically\,\footnote{\,Recall that a player is \emph{deterministic} if his output is uniquely determined by his input.} on one of them, then there also exists a classical strategy that reproduces exactly the same output distribution for any given input.
\end{lemma}
\begin{proof}(Sketch)
    We briefly outline the underlying idea. Since each player behaves deterministically on one of his two possible inputs, any nonsignalling resource can only affect the players' behaviour on the other input. The key idea is that the dependence on these resources can be shifted into the shared randomness: one can precompute the outcomes that the nonsignalling correlations would have produced on that nondeterministic input. Once this is done, the players can simulate the original strategy using only shared randomness, and thus behave classically. The formal proof is provided in Appendix~\ref{app:LemmasNonlocalGames}.
\end{proof}
Even though this is not necessary for our immediate purpose, which is to prove Theorem~\ref{thm:Theorem} stated below, we mention that the same argument gives the following more general statement, which is also proved in Appendix~\ref{app:LemmasNonlocalGames}.
\begin{lemma}
    If there exists a nonsignalling strategy for a game in which each player has at most one input on which he does not behave deterministically, then there also exists a classical strategy that reproduces exactly the same output distribution for any given input.
\end{lemma}
\newpage
We now state and prove our main result.
\begin{theorem}\label{thm:Theorem}
    The main game presented in Section~\ref{subsec:TheMainGame} cannot be won using arbitrary nonsignalling bipartite correlations and multipartite shared randomness.
\end{theorem}
\begin{proof}
     Assume, by way of contradiction, that there exists a winning strategy for the main game that uses bipartite nonsignalling correlations and multipartite shared randomness. By~Lemma~\ref{lem:SR}, we may assume without loss of generality that the strategy uses no shared randomness. 

     Now, let's consider the inflation scenario presented in Figure~\ref{fig:InflationPreuve}, which is reminiscent of the simpler proof of the 3-way $p$--coin flip  demonstrated in Section~\ref{subsec:Inflation}.

    \begin{figure}[hbt]
    \centering
    \begin{subfigure}[c]{0.39\textwidth}
      \centering
      \includegraphics[width=\linewidth]{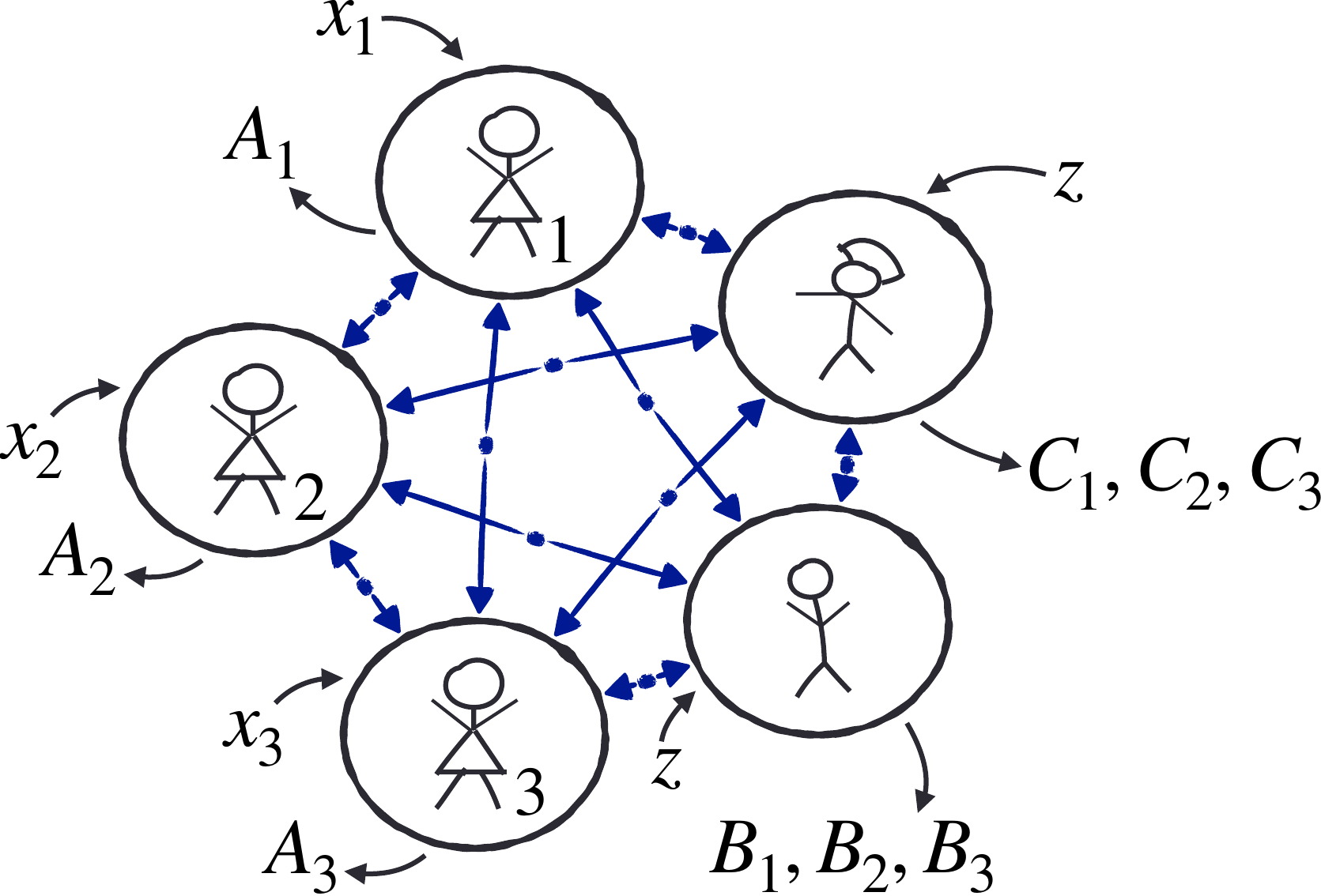}
      \caption{}
      \label{fig:Inflation2a}
    \end{subfigure}%
    \hfill
    \begin{subfigure}[c]{0.59\textwidth}
        \centering
        \includegraphics[width=\linewidth]{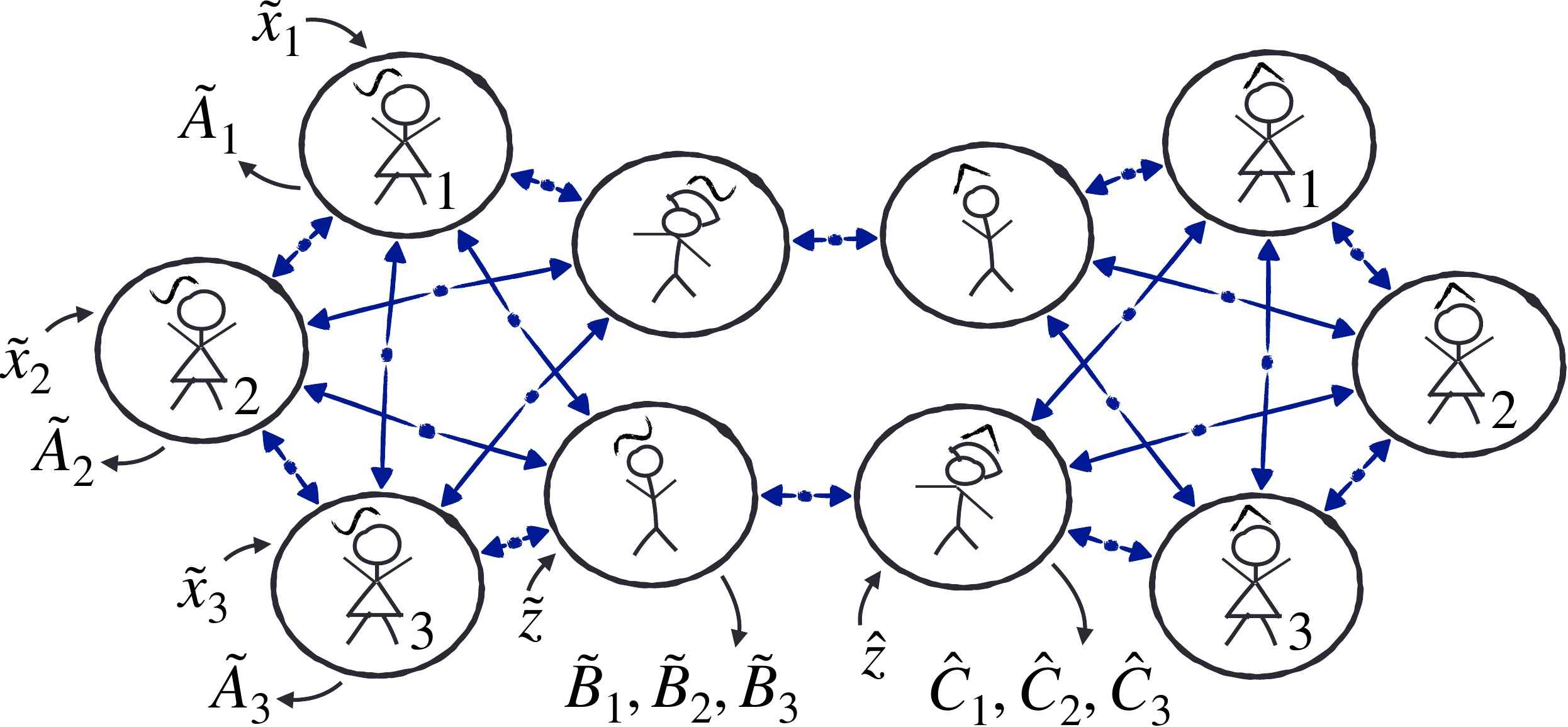}
        \caption{}
        \label{fig:Inflation2b}
    \end{subfigure}
    \caption{\small (a) Original scenario with five players ($\textnormal{Alice}_1$, $\textnormal{Alice}_2$, $\textnormal{Alice}_3$, Bob and Charlie) sharing bipartite resources pairwise. (b) Inflated scenario obtained by duplicating and rearranging players and resources. Each resource in the inflated scenario is an identical copy of the corresponding resource in the original scenario, shared between the corresponding copies of the players. For~clarity, only the inputs and outputs of players that are relevant to the argument are explicitly shown.}
    \label{fig:InflationPreuve}
    \end{figure}
    Before proceeding, we clarify our notation for the players’ outputs. For each \(i \in \{1,2,3\}\), we write $A_i(x_i)$ for the random variable corresponding to $\textnormal{Alice}_i$’s output bit when her input is $x_i$. Likewise, $B_i(z)$ and $C_i(z)$ denote the \mbox{$i$-th} output bit of Bob and Charlie, respectively, when their input is $z$. In~the inflated scenario, we use tildes and hats (e.g.~$\widetilde{A}_i(\tilde{x}_i)$, \smash{$\widehat{C}_i(\hat{z})$}) to denote the random variables of the outputs of the corresponding replicated players.  
    
    Recall that when \mbox{$z = 1$}, we have $x_1 = x_2 = x_3 = 1$, and the players must play the equality game thrice. Since the assumed strategy wins the main game, for every $i \in \{1,2,3\}$, we must have
    \begin{equation}
            \Pr[A_{i}(1) = B_{i}(1)] = 1 = \Pr[B_{i}(1) = C_{i}(1)] \,.
    \end{equation}
    On one hand, by~\ref{Fact:2}, we get
    \begin{equation}\label{eq:Alice=Bob=Charlie}
            \Pr\bigl[\widetilde{A}_{i}(1) = \widetilde{B}_{i}(1)\bigr] = 1 = \Pr\bigl[\widetilde{B}_{i}(1) = \widehat{C}_{i}(1)\bigr] \,,
    \end{equation}
    which, by transitivity, implies that \smash{\(  \Pr\bigl[\widetilde{A}_{i}(1) = \widehat{C}_{i}(1)\bigr] = 1 \)}. On the other hand, by~\mbox{\ref{Fact:1}}, for any $a \in \{0,1\}$,
    \begin{equation}
        \Pr\bigl[\widetilde{A}_{i}(1) = a, \widehat{C}_{i}(1) = a\bigr] = \Pr\bigl[\widetilde{A}_{i}(1) = a\bigr] \cdot \Pr\bigl[\widehat{C}_{i}(1) = a\bigr] \,,
    \end{equation}
    which gives us that
    \begin{align}
        1 = &\Pr\bigl[\widetilde{A}_{i}(1) = \widehat{C}_{i}(1)\bigr]\nonumber\\
        = & \Pr\bigl[\widetilde{A}_{i}(1) = 0\bigr] \cdot \Pr\bigl[\widehat{C}_{i}(1) = 0\bigr] + \Pr\bigl[\widetilde{A}_{i}(1) = 1\bigr] \cdot \Pr\bigl[\widehat{C}_{i}(1) = 1\bigr] \,.
    \end{align}
    The last equation implies the existence of a bit $a_{i}$, for each~$i$, such that $\Pr\bigl[\widetilde{A}_{i}(1) = a_{i}\bigr] = \Pr\bigl[\widehat{C}_{i}(1) = a_{i}\bigr] = 1$, and hence \smash{$\Pr\bigl[\widetilde{B}_{i}(1) = a_{i}\bigr] = 1$} as well by Eq.~(\ref{eq:Alice=Bob=Charlie}). Applying~\mbox{\ref{Fact:2}} once more, we obtain
    \begin{equation}
        \Pr[A_{i}(1) = a_{i}] = \Pr[B_{i}(1) = a_{i}] = \Pr[C_{i}(1) = a_{i}] = 1 \,.
    \end{equation}
    We conclude from this that each $\textnormal{Alice}_{i}$ always outputs $a_{i}$ on input \mbox{$x_{i}=1$}. Likewise, Bob and Charlie always output $(a_1,a_2,a_3)$ on input \mbox{$z=1$}. Hence, all five players behave deterministically on one of their two possible inputs. By~Lemma~\ref{lem:InputDeter}, we may therefore conclude that there exists a winning strategy in which all players are classical, which contradicts the fact that the main game is pseudotelepathic.
\end{proof}
Combined with the winning strategy presented in Section~\ref{subsec:TheMainGame}, Theorem~\ref{thm:Theorem} shows that the main game constitutes a separation between tripartite quantum correlations and \mbox{arbitrary} nonsignalling bipartite resources supplemented with shared randomness, within the framework of pseudotelepathy games.

\section{Conclusion}

In this work, we introduced a new pseudotelepathy game whose winning strategies fundamentally require genuine tripartite quantum entanglement. We showed that the game can be won using tripartite quantum resources, in particular, three GHZ states, while no strategy using arbitrary bipartite nonsignalling resources supplemented with shared randomness can succeed. The proof combines an inflation technique argument and a more general lemma showing that, in a nonlocal game in which each player is nondeterministic on at most one input, any nonsignalling strategy can be simulated by a classical one.

Our results provide a separation between tripartite entanglement and the strongest possible bipartite nonsignalling resources in the context of pseudotelepathy.
A natural open question is to determine the optimal winning probability achievable by strategies using arbitrary bipartite nonsignalling resources supplemented with multipartite shared randomness, since any experimental verification of the separation would necessarily involve imperfections and therefore rely on a sufficiently large gap between this value and 1.
It also remains open whether an analogous separation can be achieved with fewer players: it would be interesting to determine whether tripartite entanglement can be separated from arbitrary bipartite nonsignalling resources and multipartite shared randomness using a pseudotelepathy game with only three or four players, rather than five.

\sloppy

Furthermore, it would also be interesting to investigate whether the present result can be generalized to show a separation between \mbox{$N$-partite} entanglement and arbi\-trary \mbox{\smash{$(N-1)$}-partite} non\-sig\-nalling correlations sup\-ple\-mented with shared ran\-dom\-ness, within the frame\-work of pseudo\-telepathy games.

\fussy

\bigskip

\textbf{Acknowledgements.} We thank Scott Aaronson for valuable discussions about the majority game, including its connection to the CHSH game, and Libor Caha for suggesting that we look into Mermin's magic pentagram. We acknowledge the support of the Natural Sciences and Engineering Research Council of Canada (NSERC RGPIN-2022-04341, RGPIN-2025-05422 and ALLRP-578455-2022), Fonds de recherche du Qu\'ebec (FRQ) and Institut transdisciplinaire d'information quantique (INTRIQ).

%\printbibliography

\appendix
\newpage
\section[The GHZ--Mermin game cannot be won with arbitrary bipartite quantum resources]{The GHZ--Mermin game cannot be won with\\ arbitrary bipartite quantum resources}\label{app:GHZ-Mermin}

A significant property of the GHZ--Mermin game is its \emph{rigidity}: the winning quantum strategy, and in particular the quantum state it uses, are essentially unique. This notion was introduced by Mayers and Yao under the name of \emph{self testing}~\mbox{\cite{MayersYao2004}}, where they proved the rigidity of the optimal quantum strategy to win the CHSH game with probability $\cos^2(\nicefrac{\mathlarger{\pi}}{8})$. More to the point, as shown by Colbeck and Kent~\mbox{\cite{ColbeckKent2011}}, if Alice, Bob and Charlie win the GHZ--Mermin game with the shared state $\ket{\psi}$ on $\mathcal{H}_{A} \otimes \mathcal{H}_{B} \otimes \mathcal{H}_{C}$, then there exist local isometries
\begin{equation*}
    \Phi_{A} : \mathcal{H}_{A} \rightarrow \mathbb{C}^{2} \otimes \mathcal{H}'_{A} \,,   \Phi_{B} : \mathcal{H}_{B} \rightarrow \mathbb{C}^{2} \otimes \mathcal{H}'_{B} \,,   \Phi_{C} : \mathcal{H}_{C} \rightarrow \mathbb{C}^{2} \otimes \mathcal{H}'_{C} \,,
\end{equation*}
and some state $\ket{\textnormal{junk}}$ on $\mathcal{H}'_{A} \otimes \mathcal{H}'_{B} \otimes \mathcal{H}'_{C}$ such that
\begin{equation*}
    (\Phi_A \otimes \Phi_B \otimes \Phi_C) \ket{\psi} = \ket{\text{GHZ}} \otimes \ket{\textnormal{junk}} \, .
\end{equation*}

Furthermore, the GHZ state cannot be generated amongst three noncommunicating participants with only bipartite quantum resources~\mbox{\cite{KraftDesignolleRitzBrunnerGuhneHuber2021}}. Combining these two facts, we conclude that the GHZ--Mermin game cannot be won using arbitrary bipartite quantum resources.

\section{The majority game}\label{app:MajorityGame}

Recall that the majority game is a 3-player nonlocal game in which Alice, Bob and Charlie each receive an input bit \(x,y,z \in \{0,1\}\). Each player must output a bit \(a,b,c \in \{0,1\}\) such that \(a \oplus b \oplus c = \text{Maj}(x,y,z)\), where \(\text{Maj}(x,y,z) = 1\) if at least two of \(x,y,z\) are equal to 1, and \(\text{Maj}(x,y,z) = 0\) otherwise.

\begin{theorem}
    No quantum strategy can win the majority game with probability greater than $\cos^{2}(\nicefrac{\mathlarger{\pi}}{8})$.
\end{theorem}
\begin{proof}
    For the two notions of winning probability considered here, any strategy for the majority game that wins with probability~$p$ can be transformed into one that wins the CHSH game with probability at least $p$, which is not possible for $p > \cos^{2}(\nicefrac{\mathlarger{\pi}}{8})$, according to Tsirelson's bound \mbox{\cite{Tsirelson1980}}. We now describe these (black-box) reductions explicitly, first for the worst-case winning probability and then for the average-case winning probability under uniform input distribution.

    Recall that the CHSH game is a 2-player game in which $\textnormal{Alice}'$ and $\textnormal{Bob}'$ each receive an input bit $x',y' \in \{0,1\}$ and must output a bit $a',b' \in \{0,1\}$ satisfying $a' \oplus b' = x' \wedge y'$.

    Assume Alice, Bob and Charlie win the majority game with worst-case probability $p$. That is, they have access to three boxes that, on inputs $x$, $y$ and $z$, respectively, output $a$, $b$ and $c$, respectively, satisfying
    \begin{equation*}
        \underset{x,y,z}{\min}\Pr[a \oplus b \oplus c = \text{Maj}(x,y,z)] = p \,.
    \end{equation*}
    In an instance of the CHSH game, $\textnormal{Alice}'$ receives $x'$ and $\textnormal{Bob}'$ receives $y'$. $\textnormal{Alice}'$ inputs $x = x'$ and $z = 0$ in Alice's and Charlie's majority game's boxes, whereas $\textnormal{Bob}'$ inputs $y = y'$ in Bob's. $\textnormal{Alice}'$ obtains $a$ and $c$ from Alice's and Charlie's boxes, and outputs $a' = a \oplus c$. $\textnormal{Bob}'$ obtains $b$ from Bob's box, and outputs $b' = b$. The probability that they win their CHSH instance is given by
    \begin{align*}
        \Pr[a' \oplus b' = x' \wedge y'] &= \Pr[(a \oplus c) \oplus b = x \wedge y]\\
        &= \Pr[a \oplus b \oplus c = x \wedge y]  \\
        &= \Pr[a \oplus b \oplus c = \text{Maj}(x,y,0)] \geq p \,.
    \end{align*}
    Hence, they win the CHSH game with probability at least $p$.

    In the case that the winning probability is taken over the uniform distribution on the inputs, Alice, Bob and Charlie's boxes for the majority game verify
    \begin{equation*}
        \sum_{x,y,z} \frac{1}{8} \Pr[a \oplus b \oplus c = \text{Maj}(x,y,z)] = p \,,
    \end{equation*}
    which can be rewritten as
    \begin{equation*}
        \frac{1}{2} \sum_{x,y} \frac{1}{4} \Pr[a \oplus b \oplus c = \text{Maj}(x,y,0)] + \frac{1}{2} \sum_{x,y} \frac{1}{4} \Pr[a \oplus b \oplus c = \text{Maj}(x,y,1)] = p \,.
    \end{equation*}
    $\textnormal{Alice}'$ and $\textnormal{Bob}'$ supplement their majority game boxes with shared randomness. Since the shared randomness can be simulated by measuring a shared quantum state, this does not grant the reduction any additional power. Concretely, they sample a shared random variable $\Lambda = \{0,1\}$ with $\Pr[\Lambda = 0] = \Pr[\Lambda = 1] = \nicefrac{1}{2}$. On one hand, when $\Lambda = 0$, $\textnormal{Alice}'$ inputs $x = x'$ and $z = 0$ in Alice's and Charlie's majority game's boxes, whereas $\textnormal{Bob}'$ inputs $y = y'$ in Bob's. $\textnormal{Alice}'$ obtains $a$ and $c$ from Alice's and Charlie's boxes and outputs $a' = a \oplus c$. $\textnormal{Bob}'$ obtains $b$ from Bob's box, and outputs $b' = b$. On the other hand, when $\Lambda = 1$, $\textnormal{Alice}'$ inputs $x = x' \oplus 1$ and $z = 1$ in Alice's and Charlie's majority game's boxes, whereas $\textnormal{Bob}'$ inputs $y = y' \oplus 1$ in Bob's. $\textnormal{Alice}'$ obtains $a$ and $c$ from Alice and Charlie's boxes and outputs $a' = a \oplus c \oplus 1$. $\textnormal{Bob}'$ obtains $b$ from Bob's box and outputs $b' = b$. The probability that they win the CHSH game is given by
    {\allowdisplaybreaks
    \begin{align*}
        & \quad\sum_{x',y'} \frac{1}{4} \Pr[a' \oplus b' = x' \wedge y']\\
        &= \frac{1}{2} \sum_{x',y'} \frac{1}{4} \Pr[a' \oplus b' = x' \wedge y' \mid \Lambda = 0] + \frac{1}{2} \sum_{x',y'} \frac{1}{4} \Pr[a' \oplus b' = x' \wedge y' \mid \Lambda = 1]\\
        &= \frac{1}{2} \sum_{x,y} \frac{1}{4} \Pr[(a \oplus c) \oplus b = x \wedge y \mid \Lambda = 0]\\*[-1ex]
        &\hspace{9.56em}\mbox{}+ \frac{1}{2} \sum_{x,y} \frac{1}{4} \Pr[(a \oplus c \oplus 1) \oplus b = (x \oplus 1) \wedge (y \oplus 1) \mid \Lambda = 1]\\
        &= \frac{1}{2} \sum_{x,y} \frac{1}{4} \Pr[a \oplus b \oplus c = x \wedge y \mid \Lambda = 0] + \frac{1}{2} \sum_{x,y} \frac{1}{4} \Pr[a \oplus b \oplus c = x \vee y \mid \Lambda = 1]\\
        &= \frac{1}{2} \sum_{x,y} \frac{1}{4} \Pr[a \oplus b \oplus c = \text{Maj}(x,y,0)] + \frac{1}{2} \sum_{x,y} \frac{1}{4} \Pr[a \oplus b \oplus c = \text{Maj}(x,y,1)]\\
        &= p \,.
    \end{align*}}
    Hence, they win the CHSH game with probability $p$, concluding the reduction.
\end{proof}

\begin{theorem}
    A quantum strategy can win the majority game with probability $\cos^{2}(\nicefrac{\mathlarger{\pi}}{8})$.
\end{theorem}
\begin{proof}
    The quantum strategy is the following. Alice, Bob and Charlie share the quantum state
    \begin{equation*}
        \textstyle{\frac{1}{\sqrt{2}}} \ket{000} + \textstyle{\frac{1 \! - \! i}{2}} \ket{111} \, .
    \end{equation*}
    Each player, on input 1, applies a $\sqrt{Z}$ gate (which sends $\ket{0}$ to $\ket{0}$ and $\ket{1}$ to $i\ket{1}$). Afterwards, regardless of the input, each player applies a Hadamard gate and measures in the computational basis.
    The measurement outcome is then the player's output. A tedious, yet straightforward computation shows that this strategy wins with probability $\cos^{2}(\nicefrac{\mathlarger{\pi}}{8})$ for each possible input.
\end{proof}

\begin{theorem}
    A strategy using only three PR--boxes can win the majority game.
\end{theorem}
\begin{proof}
    The strategy is the following. Each pair of players shares a PR--box. Each player inputs his input bit into both of his boxes and outputs the XOR of the two outputs from his boxes. Let \(x,y,z \in \{0,1\}\) be the inputs of Alice, Bob and Charlie, respectively. Let the outputs of Alice's boxes be $a_1,a_2$, Bob's be $b_1,b_2$ and Charlie's be $c_1,c_2$. By the PR--box correlations,
    \begin{align*}
        a_1 \oplus b_1 &= x \wedge y \,,\\
        a_2 \oplus c_1 &= x \wedge z \,,\\
        b_2 \oplus c_2 &= y \wedge z \,.
    \end{align*}
    Thus,
    \begin{align*}
        (a_1 \oplus a_2) \oplus (b_1 \oplus b_2) \oplus (c_1 \oplus c_2) &= (a_1 \oplus b_1) \oplus (a_2 \oplus c_1) \oplus (b_2 \oplus c_2)\\
        &= (x \wedge y) \oplus (x \wedge z) \oplus (y \wedge z)\\
        &= \text{Maj}(x,y,z) \,.
    \end{align*}
    Hence, the players win the majority game.
\end{proof}

\section{The combination of the magic square game with the equality game can be won using a single PR--box}\label{app:MagicSquareAndEquality}

Inspired by the construction of Ref.~\mbox{\cite{CoiteuxRoyWolfeRenou2021}}, we define a nonlocal game that combines the magic square game and the equality game. The game involves three players: Alice, Bob and Charlie. Alice receives a row $x \in \{1,2,3\}$ in the magic square, while Bob receives a column $y \in \{1,2,3\}$. In addition, Bob and Charlie receive a common input bit $z$, which indicates whether Alice and Bob are to play the magic square game (when $z=0$) or they are all to play two parallel instances of the equality game (when $z=1$). The promise is that $x = 1$ whenever $z=1$. Each player must output two bits ($a_{1}, a_{2}$ for Alice, $b_{1}, b_{2}$ for Bob, and $c_{1}, c_{2}$ for Charlie) such that if $z=0$, we ignore Charlie's output and ask that Alice and Bob win their magic square game instance (the third bit is determined by the required parity, so it suffices to output two bits each) and if $z=1$, they all have to play the equality game twice (i.e.~their outputs must satisfy $a_{1} = b_{1} = c_{1}$ and $a_{2} = b_{2} = c_{2}$).

Although the game as described above does not admit a quantum winning strategy, a~slight modification suffices to make the game quantumly winnable. This modification allows us to obtain a pseudotelepathy game, since the modified game still admits no classical winning strategy. Bob is given two additional input bits $y_1$ and $y_2$, and the rule of the game is modified so that the players win the game whenever $z=0$ and \((y_1,y_2) \neq (c_1,c_2)\). Two GHZ states, each shared between Alice, Bob and Charlie, are sufficient to win this modified game. When $z=1$, Bob and Charlie measure their respective qubits in the computational basis and output their measurement outcomes. When $z=0$, Charlie measures his qubits in the Hadamard basis and outputs the resulting bits. In this case, Bob may treat his input bits $y_1, y_2$ as Charlie's outputs since the players win whenever these differ. Bob can therefore apply an appropriate local transformation to his qubits so that he and Alice share two EPR states. Alice and Bob can then win their magic square game instance with a strategy that uses two EPR states and in which Alice measures her qubits in the computational basis when given $x=1$. This choice of strategy ensures that the players also win the game when $z=1$ since Alice receives $x=1$ whenever $z=1$, even though she does not know which game they are playing. Such a strategy for the magic square game is well known \mbox{\cite{Aravind2004}}.

We now describe a strategy that wins this game using a single PR--box shared between Alice and Bob. Intuitively, this is possible because the magic square game with the first row fixed can be won with a single PR--box. This fact allows Alice to answer fixed outputs when $x=1$, while Bob and Charlie can answer the same fixed outputs when $z=1$, and still enables Alice and Bob to win the magic square game when $x \neq 1$. Concretely, we fix the first row of the magic square to be $(0,0,0)$ and the first column to be $(0,0,1)$, the remaining constraints are then equivalent to those of the CHSH game. Formally, Alice outputs $(0,0)$ when $x=1$. When $x=2$, she inputs $0$ into her PR--box to obtain bit $a_2$ and outputs $(0,a_2)$. When $x=3$, she inputs $1$ in her PR--box to obtain bit $a_2$ and outputs $(1,a_2)$. Bob, ignoring his auxiliary inputs $y_1$ and $y_2$, outputs $(0,0)$ when $z = 1$ or $y = 1$. Otherwise, he inputs $3-y$ into his PR--box to obtain bit $b_2$ and outputs $(0,b_2)$. Charlie outputs $(0,0)$ regardless of his inputs. It is immediate that the players win the game whenever $z=1$, since in that case $x=1$ and all three players output $(0,0)$. A simple calculation also shows that Alice and Bob win the magic square game whenever $z=0$. Hence, this combination of the magic square game with the equality game can be won using a single PR--box, even though it is pseudotelepathic. Therefore, it does not constitute a suitable candidate for exhibiting a separation between tripartite entanglement and arbitrary nonsignalling bipartite correlations.

\section{Lemmas on nonlocal games}\label{app:LemmasNonlocalGames}

We begin by introducing notation used throughout this section. For a general $N$-player nonlocal game, let $X_i$ and $A_i$ denote the input and output sets of the $i$-th player, respectively. The probability that the players output $(a_1,\hdots,a_N) \in A_1 \times \hdots \times A_N$ on input $(x_1,\hdots,x_N) \in X_1 \times \hdots \times X_N$ is denoted by 
\begin{equation*}
    \Pr\nolimits_{A_1 \hdots A_N \mid X_1 \hdots X_N} ( a_1, \hdots, a_N \mid x_1, \hdots, x_N) \,.    
\end{equation*}
Similarly, when the players have access to shared randomness $\lambda \in \Lambda$, sampled before the game instance begins, we write instead 
\begin{equation*}
    \Pr\nolimits_{A_1 \hdots A_N \mid X_1 \hdots X_N \Lambda} ( a_1, \hdots, a_N \mid x_1, \hdots, x_N, \lambda) \,.
\end{equation*}
For any subset of players $S \subseteq \{1,\hdots,N\}$, we use the natural shorthand $\Pr_{A_{S} \mid X_{S}}$ to denote the corresponding marginal distribution over the inputs and outputs of the players in $S$. For example, $\Pr_{A_i A_j \mid X_i X_j}(a_i,a_j \mid x_i,x_j)$ denotes the probability of players $i$ and~$j$ outputting $a_i,a_j$ on inputs $x_i,x_j$.

\begin{lemma}
    If there exists a winning strategy for a nonlocal game, then there also exists one that uses no shared randomness, but otherwise all the same resources.
\end{lemma}
\begin{proof}
    This follows easily from the fact that we can write
    \begin{equation*}
        \Pr\nolimits_{A_{1}\hdots A_{N} \mid X_{1}\hdots X_{N}}(a_{1},\hdots ,a_{N} \mid x_{1},\hdots ,x_{N}) =  \sum_{\lambda \in \Lambda} \Pr\nolimits_{\Lambda} (\lambda) \cdot \Pr\nolimits_{A_{1}\hdots A_{N} \mid X_{1}\hdots X_{N} \Lambda}(a_{1},\hdots ,a_{N} \mid x_{1},\hdots ,x_{N},\lambda) \,.
    \end{equation*}
    If the strategy is winning, then for every $\lambda$ such that $\Pr_{\Lambda}(\lambda) > 0 $, any output $(a_1,\hdots,a_N)$ with
    \begin{equation*}
        \Pr\nolimits_{A_{1}\hdots A_{N} \mid X_{1}\hdots X_{N} \Lambda}(a_{1},\hdots ,a_{N} \mid x_{1},\hdots ,x_{N},\lambda) > 0
    \end{equation*}
    must be accepted for the input $(x_1,\hdots,x_N)$. Thus, fixing any such $\lambda$ yields a strategy without shared randomness that wins the game.
\end{proof}

\newpage

\begin{lemma}\label{lem:deter2}
    If there exists a nonsignalling strategy for a game in which each player has only two possible inputs and behaves deterministically on one of them, then there also exists a classical strategy that reproduces exactly the same output distribution for any given input.
\end{lemma}
\begin{proof}
    Consider an $N$-player nonlocal game in which each player has exactly two inputs. Assume that there exists a strategy for this game such that each player behaves deterministically on one of them. We refer to this strategy as the \emph{initial} strategy. Our goal is to construct a fully classical strategy that reproduces exactly the same output distribution for every input. Since each player has two possible inputs, we may, without loss of generality, label them by $X_{i} = \{0,1\}$, for all $1 \leq i \leq N$. Moreover, for each player $i$, we assume without loss of generality that the initial strategy is deterministic on input $x_i = 1$. Given inputs $x_1,\hdots,x_N$, we denote by $j_1,\hdots,j_m$ the players who received input 0 and by $k_1,\hdots,k_n$ those who received input 1, with $m + n = N$. For notational convenience, we shall write $A_{j_{1}}^{j_{m}}$~for $A_{j_{1}} \hdots A_{j_{m}}$, $A_{k_{1}}^{k_{n}}$ for $A_{k_{1}} \hdots A_{k_{n}}$, $X_{j_{1}}^{j_{m}}$ for $X_{j_{1}} \hdots X_{j_{m}}$ and  $X_{k_{1}}^{k_{n}}$ for $X_{k_{1}} \hdots X_{k_{n}}$. Since each player $i$ is deterministic on input $x_i = 1$, there exists an output $\hat{a}_i \in A_i$ such that
    \begin{equation*}
        \Pr\nolimits_{A_i \mid X_i}(a_i \mid 1) = \delta(a_i,\hat{a}_i) \,,
    \end{equation*}
    where $\delta$ denotes the Kronecker delta. We begin by establishing a fact about the initial strategy. For any input $(x_1,\hdots,x_N) \in X_1 \times \hdots \times X_N$ and output $(a_1,\hdots,a_N) \in A_1 \times \hdots \times A_N$,
    \begin{align}
        &\Pr\nolimits_{A_{1}\hdots A_{N} \mid X_{1}\hdots X_{N}} (a_{1}, \hdots, a_{N} \mid x_{1}, \hdots, x_{N}) \label{eq:ProbA_NsachantX_N}\\
        &= \Pr\nolimits_{A_{j_{1}}\hdots A_{j_{m}} A_{k_{1}}\hdots A_{k_{n}}\mid X_{j_{1}}\hdots X_{j_{m}} X_{k_{1}}\hdots X_{k_{n}}} (a_{j_{1}},\hdots, a_{j_{m}}, a_{k_{1}}, \hdots, a_{k_{n}} \mid 0, \hdots, 0, 1, \hdots, 1) \label{eq:reorder}\\
        &= \Pr\nolimits_{A_{j_{1}}^{j_{m}} \mid X_{j_{1}}^{j_{m}} X_{k_{1}}^{k_{n}}} (a_{j_{1}},\hdots, a_{j_{m}} \mid 0, \hdots, 0, 1, \hdots, 1) \nonumber\\
        &\hspace{8em}\mbox{}\cdot \Pr\nolimits_{ A_{k_{1}}^{k_{n}}\mid A_{j_{1}}^{j_{m}}X_{j_{1}}^{j_{m}} X_{k_{1}}^{k_{n}}} (a_{k_{1}}, \hdots, a_{k_{n}}\mid a_{j_{1}},\hdots, a_{j_{m}},0,\hdots,0,1,\hdots,1) \label{eq:Bayes}\\
        &= \Pr\nolimits_{A_{j_{1}}^{j_{m}} \mid X_{j_{1}}^{j_{m}}} (a_{j_{1}},\hdots, a_{j_{m}} \mid 0, \hdots, 0)\nonumber\\
        &\hspace{8em}\mbox{}\cdot \Pr\nolimits_{ A_{k_{1}}^{k_{n}}\mid A_{j_{1}}^{j_{m}}X_{j_{1}}^{j_{m}} X_{k_{1}}^{k_{n}}} (a_{k_{1}}, \hdots, a_{k_{n}}\mid a_{j_{1}},\hdots, a_{j_{m}},0,\hdots,0,1,\hdots,1) \label{eq:ns}\\
        &= \Pr\nolimits_{A_{j_{1}}^{j_{m}} \mid X_{j_{1}}^{j_{m}}} (a_{j_{1}},\hdots, a_{j_{m}} \mid 0, \hdots, 0) \cdot \prod_{\ell = 1}^{n} \delta(a_{k_{\ell}},\hat{a}_{k_{\ell}}) \,,\label{eq:deter}
    \end{align} 
    where Eq.~(\ref{eq:reorder}) is obtained by a simple reordering of the players; Eq.~(\ref{eq:Bayes}) follows from an \mbox{application} of Bayes' theorem; Eq.~(\ref{eq:ns}) follows from the nonsignalling condition applied to the left-hand term; and Eq.~(\ref{eq:deter}) is a consequence of the fact that the $i$-th player systematically outputs $\hat{a}_{i}$ on input 1. We shall make use of this equation later on in the proof.

     We now define our new classical strategy, which satisfies the lemma's statement. The strategy is denoted by $\widehat{\Pr}$, and will make use of shared randomness $\Lambda$, which is defined on $A_{1} \times \hdots \times A_{N}$. We write $\lambda \in \Lambda$ as $\lambda = (\lambda_{1}, \hdots, \lambda_{N})$ with $\lambda_{i} \in A_{i}$ for each $i$. The probability distribution on $\Lambda$ and the players' strategy are defined as follows:
    \begin{align*}
        \widehat{\Pr}_{\Lambda}(\lambda_{1},\hdots,\lambda_{N}) &:= \Pr\nolimits_{A_{1}\hdots A_{N} \mid X_{1}\hdots X_{N}} (\lambda_{1},\hdots,\lambda_{N} \mid 0, \hdots, 0) \,,\\
        \widehat{\Pr}_{A_{i} \mid X_{i} \Lambda} (a_{i} \mid 0,\lambda_{1}, \hdots \lambda_{N}) &:= \delta(a_{i}, \lambda_{i})\,,\\
        \widehat{\Pr}_{A_{i} \mid X_{i} \Lambda} (a_{i} \mid 1,\lambda_{1}, \hdots \lambda_{N}) &:= \delta(a_{i}, \hat{a}_{i}) \,.
     \end{align*}
    That is, in the new classical strategy, the shared randomness is distributed as the players' joint output when they all receive input $0$ in the initial strategy. Player $i$ outputs the $i$-th component of the shared randomness, $\lambda_i$, if $x_i = 0$, and outputs $\hat{a}_i$ if $x_i = 1$. What remains to be shown is that this strategy reproduces exactly the same output distribution for any given input as the initial strategy. We have
    \begin{align}
        &\widehat{\Pr}_{A_{1}\hdots A_{N} \mid X_{1}\hdots X_{N}} (a_{1}, \hdots, a_{N} \mid x_{1}, \hdots, x_{N}) \\
        &= \sum_{\lambda_{1},\hdots,\lambda_{N}}\widehat{\Pr}_{\Lambda}(\lambda_{1},\hdots,\lambda_{N}) \cdot \widehat{\Pr}_{A_{1}\hdots A_{N} \mid X_{1}\hdots X_{N} \Lambda} (a_{1}, \hdots, a_{N} \mid x_{1}, \hdots, x_{N}, \lambda_{1},\hdots, \lambda_{N})\label{eq:SR} \\
        &= \sum_{\lambda_{1},\hdots,\lambda_{N}} \Pr\nolimits_{A_{1}\hdots A_{N} \mid X_{1}\hdots X_{N}} (\lambda_{1},\hdots,\lambda_{N} \mid 0, \hdots, 0) \cdot \prod_{t=1}^{m} \delta(a_{j_{t}},\lambda_{j_{t}}) \cdot \prod_{\ell=1}^{n} \delta(a_{k_{\ell}},\hat{a}_{k_{\ell}})\label{eq:defPhat} \\
        &= \hspace{-0.4em}\sum_{\lambda_{k_{1}},\hdots,\lambda_{k_{n}}} \hspace{-0.4em}\Pr\nolimits_{A_{j_{1}}^{j_{m}} A_{k_{1}}^{k_{n}} \mid X_{j_{1}}^{j_{m}} X_{k_{1}}^{k_{n}}} ( a_{j_{1}}, \hdots, a_{j_{m}}, \lambda_{k_{1}}, \hdots, \lambda_{k_{n}} \mid 0, \hdots, 0,0,\hdots,0) \cdot \prod_{\ell=1}^{n} \delta(a_{k_{\ell}},\hat{a}_{k_{\ell}}) \label{eq:delta}\\
        &= \Pr\nolimits_{A_{j_{1}}^{j_{m}} \mid X_{j_{1}}^{j_{m}} X_{k_{1}}^{k_{n}}} ( a_{j_{1}}, \hdots, a_{j_{m}} \mid 0, \hdots, 0,0,\hdots,0) \cdot \prod_{\ell=1}^{n} \delta(a_{k_{\ell}},\hat{a}_{k_{\ell}})\label{eq:probtot}\\
        &= \Pr\nolimits_{A_{j_{1}}^{j_{m}} \mid X_{j_{1}}^{j_{m}}} ( a_{j_{1}}, \hdots, a_{j_{m}} \mid 0, \hdots, 0) \cdot \prod_{\ell=1}^{n} \delta(a_{k_{\ell}},\hat{a}_{k_{\ell}})\label{eq:ns2}\\
        &= \Pr\nolimits_{A_{1}\hdots A_{N} \mid X_{1}\hdots X_{N}} (a_{1}, \hdots, a_{N} \mid x_{1}, \hdots, x_{N}) \,, \label{eq:PeqPhat}
    \end{align}
    where Eq.~(\ref{eq:SR}) follows directly from the definition of shared randomness; Eq.~(\ref{eq:defPhat}) uses the definition of the new strategy $\widehat{\Pr}$, where, by construction, player $i$ outputs $\lambda_{i}$ on input 0 and $\hat{a}_{i}$ on input 1; Eq.~(\ref{eq:delta}) follows from the observation that only terms satisfying $\lambda_{j_{t}} = a_{j_{t}}$ for $t \in \{1,\hdots,m\}$ have nonzero contribution; Eq.~(\ref{eq:probtot}) is obtained by summing over all possible outputs of players $k_{1},\hdots,k_{n}$; Eq.~(\ref{eq:ns2}) follows from the nonsignalling condition applied to players $j_1,\hdots,j_m$; and finally, Eq.~(\ref{eq:PeqPhat}) is an immediate consequence of Eq.~(\ref{eq:deter}) being equal to Eq.~(\ref{eq:ProbA_NsachantX_N}).
    
    Therefore, the new classical strategy reproduces exactly the output distribution of the initial strategy, which concludes the proof.
\end{proof}
\newpage
\begin{lemma}
    If there exists a nonsignalling strategy for a game in which each player has at most one input on which he does not behave deterministically, then there also exists a classical strategy that reproduces exactly the same output distribution for any given input.
\end{lemma}
\begin{proof}
    Following the lines of the proof of Lemma~\ref{lem:deter2}, consider an $N$-player nonlocal game. Assume that there exists a nonsignalling strategy for this game such that each player has at most one input on which he does not behave deterministically. We refer to this strategy as the \emph{initial} strategy. Our goal is to construct a fully classical strategy that reproduces exactly the same output distribution for every input.
    
    If all players are deterministic on all inputs, the conclusion is immediate. Otherwise, without loss of generality, the input sets are $X_i = \{0,1,\hdots,d_i\}$ for $1 \leq i \leq N$.  Moreover, also without loss of generality, we relabel the players and the inputs so that there exists \mbox{$D \in \{1,\hdots,N\}$} such that players $1,\hdots,D$ are nondeterministic on input 0 and are deterministic on inputs different from 0, while players \mbox{$D+1,\hdots,N$} are deterministic on all their inputs. Given inputs $x_1,\hdots,x_{D}$ for players $1,\hdots,D$, we denote by $j_1,\hdots,j_m$ the players who received input 0 and by $k_1,\hdots,k_n$ those who received an input different from~$0$, with $m + n = D$. For notational convenience, we shall write $A_{j_{1}}^{j_{m}}$ for $A_{j_{1}} \hdots A_{j_{m}}$, $A_{k_{1}}^{k_{n}}$ for $A_{k_{1}} \hdots A_{k_{n}}$, $A_{D+1}^{N}$ for $A_{D+1}\hdots A_{N}$, $X_{j_{1}}^{j_{m}}$ for $X_{j_{1}} \hdots X_{j_{m}}$, $X_{k_{1}}^{k_{n}}$ for $X_{k_{1}} \hdots X_{k_{n}}$ and $X_{D+1}^{N}$ for $X_{D+1}\hdots X_{N}$. By definition, when player $i$ is deterministic on input $x_i$, there exists an output $\hat{a}_i^{(x_i)} \in A_i$ such that
    \begin{equation*}
        \Pr\nolimits_{A_i \mid X_i}(a_i \mid x_i) = \delta(a_i,\hat{a}_i^{(x_i)}) \,,
    \end{equation*}
    where $\delta$ denotes the Kronecker delta. We begin by establishing a fact about the initial strategy. For any input $(x_1,\hdots,x_N) \in X_1 \times \hdots \times X_N$ and output $(a_1,\hdots,a_N) \in A_1 \times \hdots \times A_N$,
    {\footnotesize\begin{align}
        &\Pr\nolimits_{A_{1}\hdots A_{N} \mid X_{1}\hdots X_{N}} (a_{1}, \hdots, a_{N} \mid x_{1}, \hdots, x_{N}) \label{eq:ProbA_NsachantX_N_v2}\\
        &= \Pr\nolimits_{A_{j_{1}}^{j_{m}} A_{k_{1}}^{k_{n}} A_{D+1}^{N}\mid X_{j_{1}}^{j_{m}} X_{k_{1}}^{k_{n}} X_{D+1}^{N}} (a_{j_{1}},\hdots, a_{j_{m}}, a_{k_{1}}, \hdots, a_{k_{n}}, a_{D+1}, \hdots, a_{N} \mid 0, \hdots, 0, x_{k_1}, \hdots, x_{k_n}, x_{D+1},\hdots,x_{N}) \label{eq:reorder_v2}\\
        &= \Pr\nolimits_{A_{j_{1}}^{j_{m}} \mid X_{j_{1}}^{j_{m}} X_{k_{1}}^{k_{n}} X_{D+1}^{N}} (a_{j_{1}},\hdots, a_{j_{m}} \mid 0, \hdots, 0, x_{k_1}, \hdots, x_{k_n}, x_{D+1},\hdots,x_{N})  \nonumber\\
        & \hspace{1em}\mbox{} \cdot \Pr\nolimits_{ A_{k_{1}}^{k_{n}} A_{D+1}^{N}\mid A_{j_{1}}^{j_{m}}X_{j_{1}}^{j_{m}} X_{k_{1}}^{k_{n}} X_{D+1}^{N}} (a_{k_{1}}, \hdots, a_{k_{n}}, a_{D+1}, \hdots, a_{N}\mid a_{j_{1}},\hdots, a_{j_{m}},0, \hdots, 0, x_{k_1}, \hdots, x_{k_n}, x_{D+1},\hdots,x_{N}) \label{eq:Bayes_v2}\\
        &= \Pr\nolimits_{A_{j_{1}}^{j_{m}} \mid X_{j_{1}}^{j_{m}}} (a_{j_{1}},\hdots, a_{j_{m}} \mid 0, \hdots, 0)  \nonumber\\
        & \hspace{1em}\mbox{} \cdot \Pr\nolimits_{ A_{k_{1}}^{k_{n}} A_{D+1}^{N}\mid A_{j_{1}}^{j_{m}}X_{j_{1}}^{j_{m}} X_{k_{1}}^{k_{n}} X_{D+1}^{N}} (a_{k_{1}}, \hdots, a_{k_{n}}, a_{D+1}, \hdots, a_{N}\mid a_{j_{1}},\hdots, a_{j_{m}},0, \hdots, 0, x_{k_1}, \hdots, x_{k_n}, x_{D+1},\hdots,x_{N}) \label{eq:ns_v2}\\
        &= \Pr\nolimits_{A_{j_{1}}^{j_{m}} \mid X_{j_{1}}^{j_{m}}} (a_{j_{1}},\hdots, a_{j_{m}} \mid 0, \hdots, 0)  \cdot \prod_{\ell = 1}^{n} \delta(a_{k_{\ell}},\hat{a}_{k_{\ell}}^{(x_{k_{\ell}})}) \cdot \prod_{r=D+1}^{N} \delta(a_{r},\hat{a}_{r}^{(x_r)}) \,, \label{eq:deter_v2}
    \end{align}}%
    where Eq.~(\ref{eq:reorder_v2}) is obtained by a simple reordering of the players; Eq.~(\ref{eq:Bayes_v2}) follows from an application of Bayes' theorem; Eq.~(\ref{eq:ns_v2}) follows from the nonsignalling condition applied to the left-hand term; and Eq.~(\ref{eq:deter_v2}) is a consequence of the fact that players $k_1,\hdots,k_n$ and $D+1,\hdots,N$ are deterministic on their given input. We shall make use of this equation later on in the proof.

     We now define our new classical strategy, which satisfies the lemma's statement. The strategy is denoted by $\widehat{\Pr}$, and will make use of shared randomness $\Lambda$, which is defined on $A_{1} \times \hdots \times A_{D}$. We write $\lambda \in \Lambda$ as $\lambda = (\lambda_{1}, \hdots, \lambda_{D})$ with $\lambda_{i} \in A_{i}$ for each $i \in \{1,\hdots,D\}$. The probability distribution on $\Lambda$ and the players' strategy are defined as follows:
    \begin{align*}
        \widehat{\Pr}_{\Lambda}(\lambda_{1},\hdots,\lambda_{D}) &:= \Pr\nolimits_{A_{1}\hdots A_{D} \mid X_{1}\hdots X_{D}} (\lambda_{1},\hdots,\lambda_{D} \mid 0, \hdots, 0) \,,\\
        \widehat{\Pr}_{A_{i} \mid X_{i} \Lambda} (a_{i} \mid 0,\lambda_{1}, \hdots \lambda_{D}) &:= \delta(a_{i}, \lambda_{i}) \text{ for } i \in \{1,\hdots,D\}\,,\\
        \widehat{\Pr}_{A_{i} \mid X_{i} \Lambda} (a_{i} \mid x_i,\lambda_{1}, \hdots \lambda_{D}) &:= \delta(a_{i}, \hat{a}_{i}^{(x_{i})}) \text{ for } i \in \{1,\hdots,D\} \text{ and } x_i \neq 0 \,,\\
        \widehat{\Pr}_{A_{i} \mid X_{i} \Lambda} (a_{i} \mid x_i,\lambda_{1}, \hdots \lambda_{D}) &:= \delta(a_{i}, \hat{a}_{i}^{(x_{i})}) \text{ for } i \in \{D+1,\hdots,N\} \,.
     \end{align*}
    That is, in the new classical strategy, the shared randomness is distributed as the joint output of players $1,\hdots,D$ when they all receive input $0$ in the initial strategy. For $i \in \{1,\hdots,D\}$, player $i$ outputs the $i$-th component of the shared randomness, $\lambda_i$, if $x_i = 0$, and outputs~$\hat{a}_{i}^{(x_{i})}$ if $x_i \neq 0$. For $i \in \{D+1,\hdots,N\}$, player $i$ outputs~$\hat{a}_{i}^{(x_{i})}$ on input $x_i$. What remains to be shown is that this strategy reproduces exactly the same output distribution for any given input as the initial strategy. We have
    \begin{align}
        &\widehat{\Pr}_{A_{1}\hdots A_{N} \mid X_{1}\hdots X_{N}} (a_{1}, \hdots, a_{N} \mid x_{1}, \hdots, x_{N}) \\
        &= \sum_{\lambda_{1},\hdots,\lambda_{D}}\widehat{\Pr}_{\Lambda}(\lambda_{1},\hdots,\lambda_{D}) \cdot \widehat{\Pr}_{A_{1}\hdots A_{N} \mid X_{1}\hdots X_{N} \Lambda} (a_{1}, \hdots, a_{N} \mid x_{1}, \hdots, x_{N}, \lambda_{1},\hdots, \lambda_{D})\label{eq:SR_v2} \\
        &= \sum_{\lambda_{1},\hdots,\lambda_{D}} \Pr\nolimits_{A_{1}\hdots A_{D} \mid X_{1}\hdots X_{D}} (\lambda_{1},\hdots,\lambda_{D} \mid 0, \hdots, 0) \nonumber\\
        &\mbox{}\hspace{12em}\cdot \prod_{t=1}^{m} \delta(a_{j_{t}},\lambda_{j_{t}}) \cdot \prod_{\ell=1}^{n} \delta(a_{k_{\ell}},\hat{a}_{k_{\ell}}^{(x_{k_\ell})}) \cdot \prod_{r=D+1}^{N} \delta(a_{r},\hat{a}_{r}^{(x_{r})})\label{eq:defPhat_v2} \\
        &= \hspace{-0.4em}\sum_{\lambda_{k_{1}},\hdots,\lambda_{k_{n}}} \hspace{-0.4em}\Pr\nolimits_{A_{j_{1}}^{j_{m}} A_{k_{1}}^{k_{n}} \mid X_{j_{1}}^{j_{m}} X_{k_{1}}^{k_{n}}} ( a_{j_{1}}, \hdots, a_{j_{m}}, \lambda_{k_{1}}, \hdots, \lambda_{k_{n}} \mid 0, \hdots, 0,0,\hdots,0) \nonumber\\
        &\mbox{} \hspace{18.3em}\cdot \prod_{\ell=1}^{n} \delta(a_{k_{\ell}},\hat{a}_{k_{\ell}}^{(x_{k_\ell})}) \cdot \hspace{-0.4em} \prod_{r=D+1}^{N} \delta(a_{r},\hat{a}_{r}^{(x_{r})}) \label{eq:delta_v2}\\
        &= \Pr\nolimits_{A_{j_{1}}^{j_{m}} \mid X_{j_{1}}^{j_{m}} X_{k_{1}}^{k_{n}}} ( a_{j_{1}}, \hdots, a_{j_{m}} \mid 0, \hdots, 0,0,\hdots,0) \cdot \prod_{\ell=1}^{n} \delta(a_{k_{\ell}},\hat{a}_{k_{\ell}}^{(x_{k_\ell})}) \cdot \prod_{r=D+1}^{N} \delta(a_{r},\hat{a}_{r}^{(x_{r})})\label{eq:probtot_v2}\\
        &= \Pr\nolimits_{A_{j_{1}}^{j_{m}} \mid X_{j_{1}}^{j_{m}}} ( a_{j_{1}}, \hdots, a_{j_{m}} \mid 0, \hdots, 0) \cdot \prod_{\ell=1}^{n} \delta(a_{k_{\ell}},\hat{a}_{k_{\ell}}^{(x_{k_\ell})}) \cdot \prod_{r=D+1}^{N} \delta(a_{r},\hat{a}_{r}^{(x_{r})})\label{eq:ns2_v2}\\
        &= \Pr\nolimits_{A_{1}\hdots A_{N} \mid X_{1}\hdots X_{N}} (a_{1}, \hdots, a_{N} \mid x_{1}, \hdots, x_{N}) \,, \label{eq:PeqPhat_v2}
    \end{align}
    where Eq.~(\ref{eq:SR_v2}) follows directly from the definition of shared randomness; Eq.~(\ref{eq:defPhat_v2}) uses the definition of the new strategy $\widehat{\Pr}$ described above; Eq.~(\ref{eq:delta_v2}) follows from the observation that only terms satisfying $\lambda_{j_{t}} = a_{j_{t}}$ for $t \in \{1,\hdots,m\}$ have nonzero contribution; Eq.~(\ref{eq:probtot_v2}) is obtained by summing over all possible outputs of players $k_{1},\hdots,k_{n}$; Eq.~(\ref{eq:ns2_v2}) follows from the nonsignalling condition applied to players $j_1,\hdots,j_m$; and finally, Eq.~(\ref{eq:PeqPhat_v2}) is an immediate consequence of Eq.~(\ref{eq:deter_v2}) being equal to Eq.~(\ref{eq:ProbA_NsachantX_N_v2}).
    
    Therefore, the new classical strategy reproduces exactly the output distribution of the initial strategy, which concludes the proof.
\end{proof}

\printbibliography

\end{document}